\documentclass[11pt]{article} 
\usepackage{amssymb,latexsym,amsthm,amsopn,amstext,amscd,mathrsfs,amsmath}
\usepackage{enumitem}
\usepackage{amsfonts}

\usepackage[english]{babel}

\usepackage[margin=3.4cm]{geometry}

\usepackage[all]{xy}
\usepackage {graphicx}

\usepackage{authblk}

\theoremstyle{definition} 
\newtheorem{dfn}{Definition}

\theoremstyle{plain}
\newtheorem{tma}{Theorem}
 
\newtheorem{lma}{Lemma}
\newtheorem{pro}{Proposition} 

\newtheorem{cor}{Corollary}
\newtheorem{remark}{Remark}

\theoremstyle{remark}

\begin{document}
 
\title{Observable currents and a covariant Poisson algebra of physical observables}

\author[1]{	Homero G. D\'\i az-Mar\'\i n%
\footnote
{email: {\tt hdiaz@umich.mx}}
}
\author[2]{	Jos\'e A. Zapata%
\footnote
{email: {\tt zapata@matmor.unam.mx}}
}
\affil[1]{
		Facultad de Ciencias F\'\i sico-Matem\'aticas\\
		Universidad Michoacana de San Nicol\'as de Hidalgo\\
		Ciudad Universitaria C.P. 58060\\
		Morelia M\'exico\\
	}
\affil[2]{
Centro de Ciencias Matem\'aticas \\
Universidad Nacional Aut\'onoma de M\'exico \\
Morelia Mich. M\'exico, and \\
Department of Applied Mathematics \\
University of Waterloo, Waterloo, \indent Ontario, Canada	}

\maketitle

\begin{abstract}
Observable currents are locally defined  gauge invariant 
conserved currents; physical observables may be calculated integrating 
them on appropriate hypersurfaces. Due to the conservation law the hypersurfaces become irrelevant up to homology, 
and the main objects of interest become the observable currents them selves. 
Gauge inequivalent solutions can be distinguished by means of observable currents. 
With the aim of modeling spacetime local physics, we work on spacetime domains 
$U\subset M$ which may have boundaries and corners. 
Hamiltonian observable currents are those satisfying 
${\sf d_v}F=-\iota_V\Omega_L+{\sf d_h}\sigma^F$ and a certain boundary condition. 
The family of Hamiltonian observable currents is endowed with a bracket that gives it a structure which generalizes a Lie algebra in which the Jacobi relation is modified by the presence of a boundary term. 
If the domain of interest has no boundaries the resulting algebra of observables is a Lie algebra. 
In the resulting framework algebras of observable currents are associated to bounded domains, and the local algebras obey interesting gluing properties. 
These results are due 
to considering currents that defined only locally in field space 
and 
to a revision of the concept of gauge 
invariance in bounded spacetime domains. 
A perturbation of the field on a bounded spacetime domain is regarded as gauge if: 
(i) the ``first order holographic imprint'' that it leaves in any hypersurface locally splitting a spacetime domain into two 
subdomains is negligible according to the linearized gluing field equation, and 
(ii) the perturbation vanishes at the boundary of the domain. 
A current is gauge invariant if the variation in them induced by any gauge perturbation vanishes up to boundary terms. 
\end{abstract}

%{\bf Keywords:} 

%{\bf MSC:}

%{\bf Subject classification:} 

%%%%%%%%%%%%%%%%%%%%%%%%%%%%%%%%%%%%%%%%%%%%%%%%%%%%%%%%%%%%%%%%%%%%%%%%%%%%%%%%%%%%%%%%%%%%%%%%%%%%%%%%%%%%%%%%%%%%%%%%%%%%%%%%%%%%%%%%%%%%%%%%%%%%%%%%%%%%%%%%%%%%%%%%%%%

%%%%%%%%%%%%%%%%%%%%%%%%%%%%%%%%%%%%%%%%%%%%%%%%%%%%%%%%%%%%%%%%%%%%%%%%%%%%%%%%%%%%%

%%%%%%%%%%%%%%%%%%%%%%%%%%%%%

\section{Motivation}

%%%%%%%%%%%%%%%%%%%%%%%%%%%%

% NEW
%
%
The multisymplectic approach to classical field theory 
(see for example \cite{zuckerman1987action, gotay1998momentum, carinena1991multisymplectic}) 
encodes the symplectic structure present in the space of gauge equivalence classes of solutions of a classical field theory by means of a local object in the jet bundle subject to a conservation law: The pre-multisymplectic form, which may also be called the pre-symplectic current. 
One can also describe the structure in a way that brings to the forefront its relation to Topological Quantum Field Theories: 
Each hypersurface, more precisely each codimension one cycle, is assigned a space of boundary data dressed with a pre-symplectic structure. 
Each spacetime region with boundary is assigned a 
partial differential equation, the field equation, which is seen as a compatibility condition among 
boundary data on the connected components of its boundary. 
In this image, 
the conservation law is seen as the statement that for any spacetime region with boundary the space of compatible boundary data, 
according to the equation associated to the bulk, 
is a Lagrangian subspace of the space of boundary data assigned to its boundary.

%The key ingredient in this structure is a conservation law called {\em the multisymplectic formula}. 
In the Lagrangian framework 
for first order field theories 
that we will use in this article the pre-multisymplectic form will be denoted by $\Omega_L$ and when integrated on a hypersurface $\Sigma$ it yields a closed two form $\omega_{L \Sigma}$ on the space of first order data on $\Sigma$. 
The conservation law is called the 
{\em multisymplectic formula}, and it says that, when the history under consideration $\phi$ is a solution, 
given any two physical perturbations of the field $v,w$ (which are parametrized by evolutionary 
vector fields in the space of first order data satisfying the linearized field equation 
$V , W\in {\mathfrak F}$) 
we have 
\begin{equation}
\label{MSf}
\omega_{L \Sigma}(v , w) = 
\int_\Sigma j \phi^\ast \iota_{jW} \iota_{jV} \Omega_L =
\omega_{L \Sigma'}(v , w) 
\end{equation}
for any 
$\Sigma' = \Sigma + \partial U'$ for some region inside of the domain of interest $U'\subset U$.%
\footnote{ 
The details in our notation will be given in the next section. 
Now it is enough to say that 
$j\phi^\ast$ pulls back a differential form from 
a bundle to its base and 
that this differential form is constructed as the insertion of 
vector fields $jV, jW$ modeling the physical perturbations $v, w$ into $\Omega_L$. 
}
The multisymplectic approach to field theory recognizes the spacetime local object 
\[
\Omega_L
\]
as the carrier of geometric structure and brings it to the forefront. 

%Observable currents are conserved gauge invariant currents; physical observables may be calculated integrating them on 
%appropriate hypersurfaces. 
%Due to the conservation law the hypersurfaces become irrelevant up to homology, and the main objects of interest become 
%the observable currents them selves. 
%%
In a similar way, it is natural to be interested in functions, $f_\Sigma$, of boundary data on hypersurfaces that arise from a spacetime local object $F$ that is subject to a conservation law stating that when the history under consideration $\phi$ is a solution then 
\begin{equation}
\label{OCf}
f_\Sigma [\phi] = \int_\Sigma j \phi^\ast F = f_{\Sigma'} [\phi] . 
\end{equation}

The main objective of this article is the study of 
conserved currents of this type, that furthermore are gauge invariant. In Section \ref{OCsSection} we introduce them 
and call them {\em observable currents}. 
%
% Justificar la ausencia de ejemplos de observables explí­citos interesantes. Pero implícitamente tenemos suficientes. 
The explicit knowledge of a rich enough family of physical observables in a nonlinear field theory is as hard a problem as the explicit knowledge of all the solutions of that theory. Our goal is not to explicitly construct observables, 
but to study the family of 
covariant objects that precede observables of a particular type. 
A part from the family of Noether currents there is a large family of observable currents corresponding to observable currents that generalize the notion of the symplectic product function of classical mechanics $\omega(v, w)$. 
We prove that observable currents are capable of separating solutions modulo gauge. 
This result is in sharp contrast with previous reports stating that 
in nonlinear field theories, besides Noether currents corresponding to Lagrangian symmetries, there are no interesting families of conserved currents; see for example \cite{Forger+Romero, Helein, kanatchikov1998canonical, Kijowski, Goldschmidt(1973)}. 
For a compact review of the notion of observables within multisymplectic approaches to field theory see 
\cite{vey2012multisymplectic}. 
There are several differences between our treatment and those just cited, but two of them are crucial: 
The first crucial difference is that, due to the need of properly modeling vector fields in the space of solutions, the generators are allowed to depend on partial derivatives of the field of arbitrarily high order; accordingly, also the currents that we consider may depend on the field and its partial derivatives of any order. 
In the literature we found two independent studies using the ingredients briefly described above \cite{reyes2004covariant, Vitagliano(2009)}; 
the main purpose of those studies is to 
reformulate the concepts present in the covariant phase space approach to classical field theory in 
the language of the variational bicomplex or secondary calculus. 
The second crucial difference is that we consider locally defined observables; that is, observables whose domain is not all the space of solutions, but only a certain (presumably open) domain of definition.

With the intention of modeling spacetime local physics, 
we work on spacetime domains $U\subset M$ which are allowed to have boundaries and corners. This is a key motivating element of our work and also the source of most of the new problems that we faced at the beginning stage of our work and which shaped or framework.

Working in this spacetime local context forced us to review the notion of gauge equivalence in first order Lagrangian field theory over bounded spacetime domains. 
Subsection \ref{GaugeSection} is dedicated to a detailed presentation of a definition of gauge vector fields. 
The definition is motivated from a novel point of view. 
% **** PUT THIS BACK after fixing the essay ***
%; an extended presentation with relatively few formulas is presented in 
%an essay entitled ``Gauge from holography'' \cite{GaugeFromHolography}. 

%Observable currents are endowed with a bracket. 
%When integrating over a hypersurface with no boundary, the bracket induced in the algebra of observables makes it a Poisson algebra. 
%Contrast with previous lit. 
%Cite Vitagliano
%However, the bracket for the currents yields the structure of a Lie $n$-algebra. CHECK TERMINOLOGY. 
%%
% b) Hamiltonian observable currents 
In Hamiltonian mechanics the symplectic structure dictates an association between Hamiltonian vector fields and 
functions by the formula  ${ d}f = -\iota_v \omega$. In the multisymplectic approach to classical field theory it is natural to look for a version of this relation in the jet bundle that when integrated in a hypersurface $\Sigma$ induces the mentioned relation between a Hamiltonian vector field and a physical observable. We propose 
\[
{\sf d_v} F = - \iota_{jV} \Omega_L + {\sf d_h}\sigma^F ,
\]
where the boundary term $\sigma^F$ 
does not have any effect after integration on a hypersurface with 
$\partial \Sigma \subset \partial U$. 
Our proposal is derived from the study of the geometrical structure participating in this version of classical field theory. The equation above 
is introduced in Section \ref{HVFs} after 
the appropriate notion of generators of {\em multisymplectomorphisms} is identified. 
An observable current participating in the equation given above together with an associated Hamiltonian vector field will be referred to as a {\em Hamiltonian observable current}. 

The presence of the boundary term 
makes the formula less rigid than ${\sf d_v} F = - \iota_{jV} \Omega_L$ which could be guessed as a natural generalization of the formula that appears in mechanics.  
In the absence of boundaries Hamiltonian observable currents are enough to distinguish gauge inequivalent solutions, and it is natural to conjecture that 
all observable currents are Hamiltonian. 

We define a bracket for Hamiltonian observable currents in Section \ref{BracketSection}. 
When integrating over a hypersurface with 
$\partial \Sigma \subset \partial U$, observable currents lead to observables, and 
if $\Sigma$ has no boundary 
the bracket makes the space of such observables into a Lie algebra. 
The bracket among the currents 
turns out not to be a Lie bracket because the Jacobi relation is modified by a boundary term. The resulting structure is a Lie $n$-algebra 
\cite{rogers2012algebras,baez2010categorified}. 

In the resulting framework algebras of observable currents are associated to local domains; 
in Section \ref{NestedAndGluedDomains} we study the properties of local algebras corresponding to 
nested and glued domains. 

% Variational bicomplex used.
The framework used in this article uses tools and notations from the variational bicomplex. 
Since we restrict to first order Lagrangian densities, 
all the core ingredients of the framework live in the first and second jet bundles; 
however, considering currents depending on partial derivatives of the field of arbitrarily high order 
is essential for our main results. 
For the convenience of the reader we include an appendix with the minimal set of definitions needed to 
read the article. A very good brief introduction can be found in \cite{anderson1992introduction}.

The example of the Maxwell field is presented in a minimalistic style in Section \ref{Ex}. 
We provide all the necessary elements for the reader to go through the calculations by herself or himself with the 
intention of providing a familiar example that the interested reader can use to 
work out each aspect of the formalism without significant effort.

%%%%%%%%%%%%%%%%%%%%%%%%%%%%%%%%%%%%%%%%%%%%%%%%%%%%%%%%%%%%%%%%%%%%%%%%%%%%%%%%%%%%%%%%%%%%%%%%%%%%%%%%%%%%%%%%%%%%%%%%%%%%%%%%%%%%%%%%%%%%%%%%%%%%%%%%%%%%%%%%%%%%%%%%%%%

\section{General framework}
\label{Framework}

%%%%%%%%%%%%%%%%%%%%%%%%%%%%%%%%%%%%%%%%%%%%%%%%%%%%%%%%%%%%%%%%%%%%%%%%%%%%%%%%%%%%%

%
%
%

%%%%%%%%%%%%%%%%%%%%%%%%%%%%%%%%%%%%%%%%%%%%%%%%%%%%%%%%%%%%%%%%%%%%%%%%%%%%%%%%%%%%%%%%%%%%%%%%%%%%%%%%%%%%%%%%%%%%%%%%%%%%%%%%%%%%%%%%%%%%%%%%%%%%%%%%%%%%%%%%%%%%%%%%%%%

%
%

We work in a local Lagrangian first order formulation of field theory in which 
we allow domains with boundaries and corners. 
In this section we start with a 
brief review of standard material to fix notation, and 
spell out some (possibly unusual) assumptions that are essential in our framework. 
Then we  carefully review the definition of what perturbations of the field are considered to be gauge. 

% Assign phase spaces to hypersurfaces with boundary AND corners is possible if the 
% dynamics allows for a partial gauge fixation at corners (or everywhere at the boundary)

\subsection{First order Lagrangian field theory}

Histories of the field are local sections $\phi:U \subset M \rightarrow Y$ of the bundle $Y \to M$ where $U\subset M$ is a compact domain with piecewise smooth boundary. 
% CHANGE: I erased conditions stating that the domain 
% was contained inside an open subset of M. 
Physical histories are selected by Hamilton's principle according to the 
action ${S}_U (\phi) = \int_U L (j^1 \phi)$, defined by a Lagrangian density 
$L (j^1 \phi(x)) = L(x, \phi(x), D\phi(x))$ 
whose domain is the first jet bundle, $J^1Y$.

The derivative of the action in the direction prescribed by a variation of the field 
may be calculated by integration of a local object acting on the evolutionary vector field in the jet associated to the given variation. The mentioned local object is given by the variational formula  
\begin{equation}
\label{dvL}
\mathsf{d_v} { L} = {E}(L) + \mathsf{d_h} \Theta_L , 
\end{equation}
where the differential in the jet has been written as 
$\mathsf{d} = \mathsf{d_h} + \mathsf{d_v}$. 
We stress that $\mathsf{d_v} { L}$ is a differential form in 
the jet bundle%
\footnote{
Differential calculus simplifies in the 
infinite jet bundle $J Y$, the space that contains all the jets of any finite order, 
whose elements can be written as $j\phi (x) = (x, \phi(x), D\phi(x), ..., D^k\phi(x), ... )$. 
The basic differential forms in $J Y$ fit in a given jet of finite order. In first order field theory most relevant differential forms fit in the first or second jets, which are finite dimensional manifolds, but a proper treatment of variations of solutions, of gauge directions and of physical observables does involve higher order derivatives of the field. This is our reason for using 
the simplicity of calculation native to $J Y$; the reader 
may prefer to place the description in the first jet with the field equation being a differential operator and variations and observables depending on arbitrarily high derivatives of the field. 
} 
and not in the space of fields. 
The field equation is $j\phi^\ast E(L) = 0$ where 
${E}(L) = \mathsf{I} \mathsf{d_v} L$ 
is obtained from an integration by parts operator acting on $\mathsf{d_v} { L}$; the remaining term is horizontally exact 
(leading to the boundary term in the variation of the action) 
and becomes the corner stone for 
the geometric structure of this formulation of field theory.

A reader who knows a different derivation of the field equations and the geometric structure will still be able to read the paper without problems. 
For the convenience of the reader, a 
minimal set of definitions of the variational bicomplex is given in the appendix. 
In addition, the case of the Maxwell field is presented in Section \ref{Ex}. 
The intention is helping the interested reader become familiarized with this framework working on a familiar example. Thus, the last section should not necessarily be read at the end; when the reader feels the need of a more concrete explanation she or he can work it out in the example. 
A very good brief introduction to the variational bicomplex can be found in \cite{anderson1992introduction}, and 
for detailed references see for example \cite{Vinogradov, Anderson, Sardanashvily(2016)}.

Our notation for 
the space of solutions to the field equation as contained in the space of histories is 
$\mathrm{Sols}_U \subset \mathrm{Hists}_U$. 
However, we will rarely talk about the space of solutions; instead, we will often refer to the subspace 
$\mathcal{E}_{L, U} \subset J^2Y|_U$ in which ${E}(L)$ vanishes; 
% and $\mathcal{C}_L \subset J^1Y$ which is the projection of $\mathcal{E}_{L, U}$ to the first jet; 
additionally, in order to simplify notation 
the prolongation of $\mathcal{E}_{L, U}$ to higher jets 
(demanding that higher total derivatives of sections are prescribed by higher total derivatives of the field equation) will also be denoted by $\mathcal{E}_{L, U}$. 
When a local section is a solution we write 
$\phi \in \mathrm{Sols}_U$ or $j\phi (U) \subset \mathcal{E}_{L, U}$, where $j \phi(x) \in JY$ contains information regarding spacetime position together with the value of the field and all its partial derivatives evaluated at that position. In the case that we want to truncate this information up to  the first $k$ partial derivatives we would write $j^k \phi(x) \in J^kY$.

In the differential geometry of finite dimensional manifolds 
vectors in the tangent space of a given point are equivalence clases of curves passing through the point. 
Similarly, equivalence classes of one parameter families of local sections in the $k$-jet 
``passing though a given local section'' and 
consistent with a local evolution rule are {\em evolutionary vector fields} $V$ in $J^kY$; 
the subtle point about evolutionary vector fields is that the locality condition allows for the field to depend on arbitrarily high partial derivatives of the field even when $V$ lives in the $k$-jet. 
Evolutionary vector fields generate flows among sections, but in general they do not generate flows in $J^kY$. 
A geometrical picture for an evolutionary vector field $V$ can be found in its prolongation to the infinite jet $JY$ where it becomes an ordinary vector field $jV$ inducing a flow. 
We will model variations of histories by 
evolutionary vector fields such that when acting on sections preserve the property of being prolongations.%
\footnote{
See the appendix for some more information about evolutionary vector fields or \cite{Vinogradov, Anderson, Sardanashvily(2016)} 
for a detailed explanation. } 
%
%$j^1 V \in {\mathfrak X}(J^1Y|_U)$ for a vertical vector field 
%$V \in {\mathfrak X}_{\mathsf{v}}(Y|_U)$%
%\footnote{
%A vector field is called vertical if 
%at every point it is tangential to the fibers where the field takes values. 
%}. 

A perturbation $v_\phi$ (often written as $\delta \phi$) 
of a given solution $\phi$ (the tangent vector of a curve of solutions passing through $\phi$)
induces an evolutionary vector field such that $j^2V_{v_\phi}$ is 
intrinsically defined in a neighborhood inside $\mathcal{E}_{L, U}$ 
and which preserves prolongations. 
In most cases such evolutionary vector fields can be extended to a neighborhood of $\mathcal{E}_{L, U}$; 
in what follows we assume that perturbations of solutions in 
the field theory we are working with are extendible. 
%
%CONTINUE HERE
%
%
%In what follows we assume that perturbations of solutions in 
%the field theory we are working with are extendible 
%and can be represented as evolutionary vector fields in 
%neighborhoods of the second jet. 
%** CHECK Moreover, since all evolutionary vector fields in the jet are prolongations ... *** 
%we will consider prolongations of evolutionary vector fields in the first jet. 
%
The compatibility between an evolutionary vector field $V$ and the perturbation $v_\phi$ 
of a given solution involves only the restriction of $V$ to $j^1\phi(U)$; then, 
it is more precise to think that 
a perturbation corresponds to an equivalence class of evolutionary vector fields which depends on the solution; we could write $[V]_\phi$. 
The evolutionary vector fields in the class are those 
for which $j^2V|_{\mathcal{E}_{L, U}}$ preserves the property of 
sections of being prolongations and 
which satisfy the linearized field equation asking that 
$\mathscr{L}_{jV} {E}(L)|_{j^2\phi(U)}$ be horizontally exact. 
Notice that the linearized field equation is only imposed 
inside $j^2\phi(U) \subset {\mathcal{E}_{L, U}}$. 
%The space of such objects will be denoted by ${\mathfrak F}_{\phi \, U} \ni V$. 

Smooth vector fields defined in some neighborhood of the space of solutions 
assign to every solution $\phi$ in the neighborhood 
a corresponding $[V]_\phi$.
The smoothness condition implies that 
the allowed vector fields are those that 
can be associated to an equivalence class of evolutionary vector fields 
but where the equivalence relation does not depend on each particular solution. 
Allowing for evolutionary vector fields to depend on partial derivatives of the field of arbitrarily high order is essential for properly modeling vector fields in a neigborhood of solutions. 
We consider locally defined evolutionary vector fields $V$ in the first jet 
such that $j^2V|_{\mathcal{E}_{L, U}}$ preserves the property of 
sections of being prolongations and 
which satisfy the linearized field equation 
in their domain of definition. 
The equivalence relation 
defining a locally defined vector field in the space of solutions will be 
$V_1 \sim V_2$ 
if and only if 
the prolongation of their domains of definition to $J^2 Y$ intersect $\mathcal{E}_{L, U}$ in the same set and 
$j^2V_1|_{\mathcal{E}_{L, U}} = j^2V_2|_{\mathcal{E}_{L, U}}$. 
The space of such objects will be denoted by ${\mathfrak F}_U$ and even when they are equivalence classes 
of evolutionary vector fields we will denote the elements simply by $V \in {\mathfrak F}_U$. 
Each element $V$ of ${\mathfrak F}_U$ generates a flow in the space of solutions with integral curves $\phi^t$. 
Given any solution $\phi= \phi^0$ and any differential form $\sigma$ in the infinite jet we have 
$\frac{d}{dt}|_{t=0} (j\phi^t)^\ast \sigma = j\phi^\ast \mathscr{L}_{jV} \sigma$ and 
$\mathscr{L}_{[jV, jW]} \sigma = [ \mathscr{L}_{jV}, \mathscr{L}_{jW} ] \sigma$. 
It is possible to define a Lie product among evolutionary vector fields and the restriction of 
${\mathfrak F}_U$ to $\mathcal{E}_{L, U}$ turns out to be a Lie subalgebra; moreover, this structure is compatible with the prolongation to 
the infinite jet and the Lie algebra of vector fields there.% 
\footnote{
A brief explanation of the structure of ${\mathfrak F}_U$ can be found in the appendix. 
In the mathematical literature this representation of 
vector fields in the space of solutions within the jet is referred to as generalized symmetries or higher symmetries; 
for a detailed explanation see \cite{SYMM}.
} 

As mentioned earlier the most important differential forms in the geometrical structure of field theory live in the first and second jets, and even when we work in the infinite jet 
one may think that objects belong to the first or second jets. 
However, the contraction of an evolutionary vector field $V$ with 
a differential form from the first jet may depend on higher derivatives of the field. 
For example $\iota_{V} \Theta_L$
would generically be a function depending on partial derivatives of $\phi$ of order greater than one; then we will write its pull back to spacetime as 
$j \phi^\ast \iota_{jV} \Theta_L \in \Omega^{n-1}(M)$ and regard those operations as taking place in the infinite jet, which is nothing more than the organization of all the jets of different orders.

\begin{remark}[The multisymplectic point of view]
\label{MultisymplecticPV}

The conservation law for the pre-multisymplectic form $\Omega_L = -\mathsf{d_v} \Theta_L$, 
a form of vertical degree two and horizontal degree $n-1$, 
following from $\mathsf{d_v}^2 L = 0$ and the variational formula (\ref{dvL}) 
may be written as $(\mathsf{d_h}\Omega_L)|_{{\mathfrak F}_U , \mathcal{E}_{L, U}}= 0$ 
or as the {\em multisymplectic formula} written in the introduction (\ref{MSf}). 
This relation is the local corner stone of the point of view given in this remark. 

Now we present a point of view which emphasizes the aspects of this framework for classical field theory that follow a structure similar to that of topological quantum field theories. 
Once we know what are the degrees of freedom of the field under study (the typical fiber in the fiber bundle describing the field) and the Lagrangian, we can define the following assignments posing a problem in classical field theory: \\
(i) Oriented spacetime hypersurfaces $\Sigma$ possibly with boundary (and corners) locally divide spacetime into two sides and this makes them candidates to host boundary conditions. Thus, we 
associate to each oriented hypersurface $\Sigma$ a space of data that locally provides appropriate boundary conditions. Since we are considering field theories with a second order field equation, to each hypersurface $\Sigma$ we associate the space of first order data on it $J^1Y|_\Sigma$ (or the space of sections of this bundle). Then $\Omega_L$ becomes an instrument providing the mentioned spaces with $\Omega_L|_\Sigma$ (or pre-symplectic forms $\omega_{L \Sigma}$ for the space of sections). \\
(ii) On the other hand, spacetime regions $U$ possibly with boundary (and corners) have fields on them that need to satisfy the field equation; thus, to each spacetime region $U$ we associate $\mathcal{E}_{L, U}$ together with the pre-multisymplectic form $\Omega_L$ 
(or the space of solutions $\mathrm{Sols}_U$ with its pre-symplectic structure). 

The collection of hypersurfaces and regions accompanied with the structure that we just described 
is cohesive in the sense that it 
follows a compatibility condition --the multisymplectic formula (\ref{MSf})-- written as 
\[
\int_{\partial U} j \phi^\ast \iota_{jW} \iota_{jV} \Omega_L = 0
\]
for any $\phi \in \mathrm{Sols}_U$ and any pair of tangent vectors to the space of solutions modeled here by $V$ and $W$. 

These assignments serve the purpose of posing problems in classical field theory. Solving one such problem on a region $U$ means finding sections $\phi$ whose prolongations satisfy $j\phi(U) \subset \mathcal{E}_{L, U}$ and comply with the boundary conditions on the hypersurface $\partial U$. 
Additionally, 
$\Omega_L$ induces a pre-symplectic structure to the space of solutions and, as we will see in the following subsection, 
also induces the 
notion of gauge vector field leading to gauge equivalence classes of solutions.

The assignment rules stated above enjoy some gluing properties which in this formulation consist in gluing submanifolds of jet bundles. The problem of gluing solutions along a hypersurface $\Sigma$ is considered in the following subsection leading to a {\em gluing field equation}. 

It is also possible to take a dual point of view not focussing on solutions (mod gauge) at $U$, but on observables. In this work we do this by means of studying observable currents in each spacetime domain $U$. Thus, to each spacetime region $U$ we associate a space of observable currents ${\rm OC}_U$; developing the notion of an observable current is the subject of Section \ref{OCsSection}. Later on we study the gluing properties of assignment in Section \ref{NestedAndGluedDomains}. 
\end{remark}

\begin{remark}[Cohomology classes vs a local description, and corner ambiguities]
\label{cohomology}
Notice that since $\Theta_L$ arises as the boundary term in the variation of the action it is to be integrated at boundaries or connected components of boundaries which are cycles. Additionally, conservation laws (in objects derived from $\Theta_L$) 
tell us that those cycles will be relevant only modulo boundaries. 
Thus, we may regard $\Theta_L$ as relevant only up to its horizontal cohomology class. 
This phenomenon has been called a ``corner ambiguity'' in the recent literature because the ambiguity becomes relevant in the presence of corners (codimension two surfaces). 
Thus, even when we write $\Theta_L$ it may seem more appropriate to think about its horizontal cohomology class, and this remark 
applies precisely to 
the pre-multisymplectic form $\Omega_L$ (which obeys a conservation law). 
However, the resulting framework would not be appropriate to model local physics --like describing what happens in a laboratory during the course of an experiment-- 
because the objects in the framework 
would be integrable only on extended hypersurfaces that could not be split into smaller pieces; we would not be able to 
compute quantities in a local way to compare them with the measurements performed in the laboratory. Thus, we will force the framework to let us work in compact spacetime domains $U\subset M$ in such a way that integration on hypersurfaces with 
$\partial \Sigma \subset \partial U$ can be done. Two stages are needed to accomplish this goal: 
(i) Chose a representative out of the mentioned cohomology classes; this is something that we often do, and formulas for the pre-symplectic current 
(here called pre-multisymplectic form) 
$\Omega_L$ and its potential $\Theta_L$ can be found in textbooks (where the absence of corners makes the ambiguity irrelevant).  
(ii) Define a notion of gauge equivalence that is consistent with making such choices (which is the subject of the next subsection). 

Alternatively, we could have decided to work with a modified notion of horizontal cohomology classes in which horizontally exact terms are identified with zero only if they vanish in the bundle over $\partial U$ in the appropriate sense. We chose the strategy described above, but the relevance of this modified cohomology classes will be clear in various formulas. 

A separate issue is that 
Lagrangian densities leading to the same variational problem should be considered equivalent, and 
adding boundary terms $L \to L + \mathsf{d_h} b$ does not modify the 
problem stated by Hamilton's principle of least action. 
Thus, the field equation ${E}(L)$ remains invariant under the addition of boundary terms, while the horizontally exact term changes as 
$\Theta_L \to \Theta_L - \mathsf{d_v} b$ and the pre-multisymplectic form $\Omega_L$ also remains invariant. 

Hence, $\Omega_L$, modulo its inherited corner ambiguity, 
is the carrier of invariant geometrical structure in this framework 
and allowing for the multisymplectic point of view described above. 
\end{remark}

\subsection{Gauge freedom}
\label{GaugeSection}
% New starts

In physics, a description includes gauge freedom if physically distinct configurations do not correspond to points in the space that hosts it, but to equivalence classes. Often the equivalence classes are the orbits of certain vector fields declared to be gauge vector fields 
in the space of solutions of the field equation. 
In the Lagrangian first order formalism gauge freedom can be understood considering the propagation of perturbations through hypersurfaces: 
{\em Gauge perturbations have null first order holographic imprint on any hypersurface}. 
Gauge equivalence and locality have a delicate relation; our framework is phrased within confined spacetime domains 
which may be glued to other spacetime regions through shared faces in their boundary. 
Our goal is to compute physically meaningful objects for each bounded domain and be able to construct relevant objects in composite spacetime domains as appropriate compositions. 
Below we give a precise definition of the notion of gauge freedom and explain its motivation. For different arguments leading to a closely related related but inequivalent definition of gauge freedom see Wald and Lee \cite{lee1990local}. 
%
%
%NEW

\begin{dfn}[Gauge vector fields]
\label{Gauge}
$X \in {\mathfrak F}_U$, modeling a 
vector field in the space of solutions 
%\footnote{
%We recall that perturbations of a solution  as defined in the previous subsection are not ordinary vector fields in the jet but evolutionary vector fields solving the linearized field equation around that solution, and vector fields on the space of solutions are also modeled by evolutionary vector fields solving the linearized field equation in a neighborhood inside $\mathcal{E}_{L, U}$. 
%} 
is declared to be be a gauge vector field if and only if: 
\begin{enumerate}
\item
$(\iota_{jX} \; \Omega_L)|_{{\mathfrak F}_U , \mathcal{E}_{L, U}}$ coincides with a horizontally exact form. 
\item
$jX |_{\mathcal{E}_{L, U} , \partial U} = 0$; \\
i.e. 
$jX$ vanishes on the intersection of $\mathcal{E}_{L, U}$ with the sub bundle of the jet over $\partial U$. 
%
%\footnote{
%Notice that even when $X$ is an object in the first jet, we wrote its restriction to the restriction of the jet over $\partial U$ in an abbreviated fashion. We will continue to use this simplified notation throughout the paper. 
%}
\end{enumerate}
We will say that $X$ is a gauge vector field and 
write 
$X \in {\mathfrak G}_U$. \\
If Condition 1 is satisfied 
we will write $X \in \hat{\mathfrak G}_U$; clearly ${\mathfrak G}_U \subset \hat{\mathfrak G}_U \subset {\mathfrak F}_U$. 
\end{dfn}

Notice that the definition of gauge vector fields is independent of the choice of 
$\Omega_L$ among its cohomology class; in other words, it is not affected by the corner the ambiguities of the pre-multisymplectic form. 

Below we will prove that the restriction of 
${\mathfrak G}_U$ to $\mathcal{E}_{L, U}$ is a Lie subalgebra of ${\mathfrak F}_U$. 
Thus, locally this definition induces a notion of gauge equivalence classes in the space of solutions.%
\footnote{
The global existence of gauge orbits is a nontrivial issue. For a heuristic discussion of this issue and its relevance in quantization see \cite{dittrich2017can}. 
} 

If we are working on a domain of the type $U = \Sigma \times [0, 1]$ 
endowed with a foliation $\Sigma_t$ and we are interested only in integration over hypersurfaces that belong to the foliation, 
we may replace Condition 2 by $X |_{\partial \Sigma \times [0, 1]} = 0$; this will be addressed below. If one is interested in working with initial data at a Cauchy surface $\Sigma$, or a Cauchy surface that has been divided into pieces, 
one may want to think of a limit of spacetime domains of the form $\Sigma \times [0, \epsilon]$. 

Condition 1 in the definition 
has its origin in how perturbations of solutions propagate through codimension one surfaces. 
It says that 
if the perturbation induced by $X$ on any solution is such that for every hypersurface 
the corresponding first order holographic imprint is 
equivalent to the imprint of of the null perturbation 
then $X$ should be regarded as gauge. 
Now we give a more detailed explanation about gluing perturbations supporting Condition 1. 
Consider a spacetime domain $U \subset M$ and an arbitrary partition of it into two pieces separated by a hypersurface, 
$U = U_1 \cup U_2$ with $\Sigma = U_1 \cap U_2$. 
Either 
$\Sigma$ is a cycle (i.e. $\partial \Sigma = \emptyset$)
or $\partial \Sigma \subset \partial U$. 
Let us write the field as the gluing of its restriction to the pieces of the domain 
$\phi = \phi_1 \#_{\Sigma} \phi_2$, where the use of the gluing symbol assumes that the field is continuous at $\Sigma$.%
\footnote{
Continuity in the directional derivative in directions transversal to $\Sigma$ (up to gauge) will be a consequence of the gluing field equation. 
} 
The action and its variation are additive under such a subdivision of the domain, $dS_U = dS_{U_1} + dS_{U_2}$. 
However, when we split the domain in two pieces the degree of differentiability of the field over $\Sigma$ is relaxed and the 
usual variation of the action 
$dS_U [v_\phi] = \int_U j\phi^\ast \iota_{jV} E(L) + \int_{\partial U} j\phi^\ast \iota_{jV} \Theta_L$
following from (\ref{dvL}) 
acquires an extra term associated to $\Sigma$ 
\begin{equation}
\label{GlEq}
\int_{\Sigma} (j\phi_1^\ast - j\phi_2^\ast) \iota_{jV} \Theta_L 
. 
\end{equation}
If we look for extrema of $S_U$, apart from field equations at $U_1$ and $U_2$ there is a gluing field equation at $\Sigma$ requiring that the above integral vanishes for variations that vanishes at $\partial U$. 
In fact, if we first demand that the field equation holds inside $U_1$ and $U_2$ the 
gluing condition should demand only that $S_{U_1}+ S_{U_2}$ be extremal among variations that preserve solutions; the gluing field equation at $\Sigma$ demands that the above integral vanishes for variations in ${\mathfrak F}_U$ vanishing at $\partial U$. 
The local formulation of this condition at $\Sigma$ is the momentum flux matching condition 
$(j\phi_1^\ast - j\phi_2^\ast) \mathsf{I} (\Theta_L|_{{\mathfrak F}_U})|_{\Sigma} =0$, 
where we have written the field as $\phi = \phi_1 \#_{\Sigma} \phi_2$ and 
$\mathsf{I}$ is the integration by parts operator.%
\footnote{
In the appendix we recall the definition of the integration by parts operator. 
We would like to note that in the same way that we assigned $\mathcal{E}_{L, U}$ to the spacetime domain $U$, we can assign to the hypersurface $\Sigma$ the submanifold of the jet $\mathcal{E}_{L, \Sigma}$ in which the the gluing field equation vanishes. 
} 
Now consider a one parameter family of fields $\phi_t$ (with $\phi_{t=0} = \phi$ and with the variation at $t=0$ given by 
$V = V_1 \#_{\Sigma} V_2$) 
solving the field equation in $U_1$ and $U_2$ and solving the gluing problem over $\Sigma$. Since for each value of the parameter the field $\phi_t$ is an extremum of 
(\ref{GlEq}), at first order in $t$ we have 
$\mathsf{I}(\Theta_L(j(\phi + t (jV_1- jV_2))|_{{\mathfrak F}_U})|_\Sigma= 0$. Thus, the linearized gluing equation is 
\begin{equation}\label{LGEQ}
0 = j\phi^\ast \mathsf{I}( \mathscr{L}_{jV- jV_2} \Theta_L|_{{\mathfrak F}_U}) |_{\Sigma} =
-j\phi^\ast \mathsf{I}( \iota_{jV_1- jV_2} \Omega_L|_{{\mathfrak F}_U} )|_{\Sigma} , 
\end{equation}
where we have used 
$\mathscr{L}_{jX} = \iota_{jX} {\sf d_v} + {\sf d_v} \iota_{jX}$ and the fact that 
$(V_1 - V_2)|_{j\phi(\Sigma)} = 0$ implies that the vector field in the jet $jV_1- jV_2$ is in the kernel of $\Theta_L|_{j\phi(\Sigma)}$, as can be readily verified from its expression within a coordinate system. 
Since the integration by parts operator 
decomposes any $n-1$ horizontal form as 
$\mu = \mathsf{I}(\mu) + {\sf d_h} \sigma$, and it satisfies $\mathsf{I}^2= \mathsf{I}$, $\mathsf{I} {\sf d_h}= 0$, we see that the linearized gluing field equation is equivalent to requiring that 
$\iota_{j^1V_1- j^1V_2} \Omega_L|_{{\mathfrak F}_U}$, when evaluated in 
the intersection of $\mathcal{E}_{L, U}$ and 
the bundle over $\Sigma$, be equal to a horizontally exact form.

Thus, 
the requirements for gluing perturbations $V_1 \#_\Sigma V_2$ across hypersurface $\Sigma$ are: 
(C) Continuity of the perturbation at $0$th order in the jet bundle over the dividing hypersurface 
where the field equation holds, which is equivalent to 
$j_\Sigma (V_1|_{\mathcal{E}_{L, U} , \Sigma}) = j_\Sigma (V_2|_{\mathcal{E}_{L, U} , \Sigma})$. 
(LG) The linearized gluing field equation, 
that $\iota_{jV_1- jV_2} \Omega_L$ be equal to a horizontally exact form, 
as an operator acting on 
the restriction of 
evolutionary vector fields in ${{\mathfrak F}_U}$ 
to $\mathcal{E}_{L, U}|_\Sigma$. 
This equation contains a germ of information regarding the bulk; more precisely, it contains 
first order partial derivatives of the perturbations in directions transversal to the dividing hypersurface $\Sigma$. We call this information the {\em first order holographic imprint} of the perturbation. 
The linear operator which appears in the 
linearized gluing equation may have a nontrivial null space; 
such a linearized gluing equation would find the imprint left by some non zero perturbations as negligible. For those perturbations, 
propagation through a dividing hypersurface proceeds without any 
trace of bulk information. 
{\em Vector fields satisfying Condition 1 may have a complicated form in the bulk, but as far as 
propagation through $\star$any$\star$ dividing hypersurface all this information is lost; the definition of gauge vector fields declares those degrees of freedom as 
physically unimportant. 
This is the motivation for Condition 1 in the definition of gauge vector fields.} 
Further support for Condition 1 is given in Remark \ref{FurtherSup} of 
section \ref{HVFs}, where we consider the notion of multisymplectomorphisms and related locally Hamiltonian vector fields. 
In Section \ref{Ex} we show how in the case of Maxwell's field the familiar notion of gauge arises from Condition 1. 

We mentioned that since we work with a first order Lagrangian density, most differential forms in our formalism fit in the first or second jet bundles. We must warn the reader that 
even when $\Omega_L$ is a differential form fitting in the first jet, its contraction with elements in ${\mathfrak F}_U$ will in general depend on higher order derivatives of the field. 
If one insists in working in the first jet the condition of horizontal exactness, 
$(\iota_X \; \Omega_L - {\sf d_h}\rho)|_{{\mathfrak F}_U , \mathcal{E}_{L, U}} = 0$, 
should be interpreted with $\rho$ being a differential operator of arbitrarily high order. 

%THE THING WRITTEN BELOW SHOULD BE WRONG IN GENERAL. 
%Due to a technical lemma \cite{Takens(1979)} reviewed in the appendix, 
%an equivalent condition to Condition 1 
%is the equation 
%\[
%({\sf d_h}\iota_{X} \Omega_L)|_{\mathcal{E}_{L, U}} =0
%\]
%which may be simpler to verify. 

Condition 2 of the definition of gauge vector fields is essential for the integration of 
currents on hypersurfaces with $\partial \Sigma \subset \partial U$ producing gauge invariant quantities. 
Related considerations appeared long time ago in a study of the role of surface integrals in General Relativity by Regge and Teitelboim 
\cite{regge1974role}. 
Below, in Remark \ref{Ginv} we will spell out the condition on a current to be gauge invariant. From the definition it is clear that without Condition 2, demanding gauge invariance would force us to work only with cohomology classes that 
could be integrated only on cycles, which would render most allowed calculations trivial for 
confined spacetime domains $U \subset M$. 
Another reason supporting Condition 2 is our 
interest in measurements of the field relative to the field itself 
since ultimately, when all fields are considered as part of the system under study, only 
this type of measurements would be available. In this setting $\partial U$ separates the system into two subsystems, and we may want to measure properties of the field inside $U$ with respect to the field outside. Now, since in the first order formalism all information from the outside is encoded in 
$j^1\phi |_{\partial U}$ we may consider $j^1\phi |_{\partial U}$ as a reference and keeping the reference gauge invariant would be wise. 
Additionally, Remark \ref{IsolatedSystems} argues that measuring properties of 
an ``isolated system'' may need a reference at ``infinity'' 
and preserving that reference frame may be essential for talking about those properties.
Yet another reason for including Condition 2 in our definition comes from the standard definition of gauge vector fields as generators of Lagrange symmetries depending on arbitrary local parameters. 
Wald and Lee \cite{lee1990local} start with a precise version of that definition of gauge and find that it implies our Condition 1, but along their argument they assume that if there is a boundary it is located at infinity which (together with appropriate fall-off conditions on the field) 
lets them conclude that the Noether charge associated to a gauge vector field $X$ according to their definition vanishes identically $Q^X_\Sigma =0$. 
The interested reader is invited to try to reproduce the mentioned argument by Wald and Lee in the context of a domain with boundary 
using the result shown in Remark \ref{ConsCharWouldBeGauge}; Condition 2 will emerge naturally. 
Recently, Freidel and Donelly \cite{donnelly2016local} emphasized that in domains with boundary there are ``would be gauge degrees of freedom'' living at the boundary; those degrees of freedom could be understood as having origin in Condition 2; see Remark \ref{WouldBeGauge}. Their motivation came from entanglement entropy in gauge theories \cite{donnelly2014entanglement} and general relativity in spacetime domains with corners \cite{freidel2015quantum}. 

% *** Put it back after fixing it. ***
%The definition of gauge vector fields given above and 
%an expanded version of this supporting argument are presented in a essay entitled 
%{\em Gauge from holography} \cite{GaugeFromHolography}. 
%

\begin{remark}[Other definitions of gauge]
\label{OtherDefsGauge}
Several references in the context of classical field theory and the variational bicomplex give definitions
closely related to Condition 1 (see for example \cite{Sardanashvily(2016)}). 
The work of Wald and Lee is the standard reference for the subject in the context of the covariant phase space formulation of field theory \cite{lee1990local}. 
The rough physical idea behind those other definitions of gauge is 
that families of symmetries depending on locally independent parameters become an obstacle for predictability of the theory and should be regarded as gauge. 
Complementary features of that notion of gauge symmetries are that there are relations among the field equations (and the linearized field equations) and that 
the Noether currents associated to the gauge symmetries vanish on-shell (up to boundary terms); these phenomena are the content of Noether's second theorem. 
Another important property is that the evaluation of a (pre)symplectic product of variations is independent of changes of the variations in gauge directions. 

Every gauge vector field according to Wald and Lee satisfies Condition 1 of our definition \cite{reyes2004covariant, Vitagliano(2009),lee1990local}. 
The work of Reyes \cite{reyes2004covariant} states the conjecture that 
Condition 1 may imply the usual definition, and in many references a strong version of Condition 1 (not asking that the form is horizontally exact, but that it vanishes) is considered as an indication that the vector field is gauge. 
\end{remark}

%%%

\begin{remark}[Gauge invariance]
\label{Ginv}
A function of the jet $f: JY \to {\mathbb R}$ is gauge invariant if it remains constant along gauge orbits in the intersection of its domain of definition with 
${\mathcal E}_L \subset JY$; similarly, a locally defined function of histories is gauge invariant if it remains constant along gauge orbits 
in the intersection of the space of solutions 
with its domain of definition. (The local existence of the mentioned gauge orbits is justified below.) 
Since in our work currents play a central role, we need to spell out the meaning of gauge invariance for them. 
The natural gauge invariance requirement for a current is to ask that its integration on cycles produces gauge invariant functions. 
The corresponding local requirement in the jet is to call 
a current $F$ (a locally defined $n-1$ horizontal form in the jet) gauge invariant if for every $X \in {\mathfrak G}_U$ 
\begin{equation}
\label{ginv}
(\mathscr{L}_{jX} F - {\sf d_h} \sigma)|_{\mathcal{E}_{L, U}} = 0 \quad 
\mbox{ for some } \sigma . 
\end{equation}
Notice that $\sigma$ must be linear in $X$ and that Condition 2 in Definition (\ref{Gauge}) 
implies that the restriction of $\sigma|_{\mathcal{E}_{L, U}}$ to the sub bundle over $\partial U$ vanishes, 
$(\sigma|_{\mathcal{E}_{L, U}}) |_{\partial U}= 0$. 
Any spacetime cycle $\Sigma \subset M$ may be decomposed as a sum of hypersurfaces contained in compact domains $U_i$ with 
$\partial \Sigma_i \subset \partial U_i$ and we may write $f_\Sigma [\phi]$ as a sum of contributions 
$f_{\Sigma_i} (\phi) = \int_{\Sigma_i} j \phi^\ast F$. 
Due to Condition 2 each $f_{\Sigma_i} (\phi)$ is gauge invariant. 
However, if $\Sigma$ is a hypersurface with boundary and $\partial \Sigma_i$ is not contained in $\partial U$ then $f_{\Sigma} (\phi)$ 
is not gauge invariant; if we choose a representative in the cohomology class of $F$ to calculate $f_{\Sigma} (\phi)$ a gauge transformation would not preserve our choice and the resulting boundary term in the integral would not vanish. 

Gauge vector fields $X \in {\mathfrak G}_U$ preserve the pre-multisymplectic form in the sense that 
$\mathscr{L}_{jX} \Omega_L|_{\mathcal{E}_{L, U}}$ is horizontally exact. Thus, 
the pre-symplectic form obtained by integration on any cycle (after the appropriate insertions of elements of ${\mathfrak F}_U$) 
as in formula (\ref{MSf}) is gauge invariant, 
$\mathscr{L}_{jX} \omega_{L \Sigma}|_{\mathrm{Sols}_U} = 0$. Additionally, if a hypersurface is not a cycle but 
$\partial \Sigma \subset \partial U$ then $\omega_{L \Sigma}$ is also preserved by gauge transformations. 
\end{remark}

\begin{remark}[Gauge equivalence classes]
\label{GeqCl}
We need to talk about equivalence classes of solutions. 
%
%CHECK ARGUMENT. 
%
The local existence of orbits in the space of solutions 
follows from ${\mathfrak G}_U \subset {\mathfrak F}_U$, 
when restricted to $\mathcal{E}_{L, U}$ 
being a Lie subalgebra. 
Given any $X, Y \in {\mathfrak G}_U$ 
(with $(\iota_{jX} \Omega_L - {\sf d_h} \rho^X)|_{\mathcal{E}_{L, U} , {\mathfrak F}_U} = 0$, 
$\iota_{jY} \Omega_L - {\sf d_h} \rho^Y)|_{\mathcal{E}_{L, U} , {\mathfrak F}_U} = 0$) 
a short calculation yields 
\[
(\iota_{[jX, jY]} \Omega_L -  
{\sf d_h} (\mathscr{L}_{jX} \rho^Y - \iota_{jY} {\sf d_v} \rho^X ))|_{{\mathfrak F}_U , \mathcal{E}_{L, U}} = 0 . 
\]
Since $[jX, jY] = j[X, Y]$ the restriction to $\mathcal{E}_{L, U}$ of 
${\mathfrak G}_U \subset {\mathfrak F}_U$ is a Lie subalgebra. 
The calculation above proves directly that 
the generators of gauge flows in $\mathcal{E}_{L, U} \subset JY$ form a Lie algebra, and  
the space of solutions inherits a representation of this Lie algebra. 
This guarantees the local existence of gauge orbits. Given that our definition of observables will be  
based on locally defined functions the existence of gauge orbits is relevant in our framework. 
The issue of global existence of orbits is a hard problem; for a heuristic discussion of this issue and its relevance in quantization see \cite{dittrich2017can}.

% IN GENERAL THE FLOW MENTIONED BELOW DOES NOT EXIST
%Thus, the flows of these vector fields define the local action of a group, the gauge group $G$, in a neighborhood of $\mathcal{C}_L \subset J^1Y$ preserving $\mathcal{C}_L$. A more precise geometric picture is obtained looking at the field equation in the second jet $\mathcal{E}_{L, U} \subset J^2Y$ where again the gauge group acts on a neighborhood of $\mathcal{E}_{L, U}$ preserving it. 
%For heuristic arguments it 
%will be relevant to have in mind the local product structure induced by gauge equivalence. Each point of $\mathcal{E}_{L, U}$ is contained in a neighborhood $\Delta$ 
%that is decomposed as a product of gauge orbits over a space of gauge classes $(\mathcal{E}_{L, U}/G)_\Delta$ that is a bundle over $U$.  
%However, we will continue to work in $J^1Y$ (and $J^2Y$) looking for objects that are appropriately gauge invariant. 

%Vector fields $V \in {\mathfrak X}_{\mathsf{v}}(J^1Y|_U)$ that are gauge orbit preserving are also 
%invariant under the flow of gauge vector fields 
%(modulo gauge vector fields); then, we will refer to these vector fields frequently as gauge invariant. 

The subalgebra of 
${\mathfrak F}_U$ preserving ${\mathfrak G}_U$ will be denoted by 
\[
	{\mathfrak F}^{{\mathfrak G}}_U
	:=
	\left\{
		V\in{\mathfrak F}_U
		:\; \; 
		\mathscr{L}_{V} X \in {\mathfrak G}_U
		\quad
		\forall
		X\in {\mathfrak G}_U
	\right\} . 
\]
Since ${\mathfrak G}_U \subset {\mathfrak F}^{\mathfrak G}_U$ is 
a Lie ideal 
the quotient makes sense and inherits 
a Lie algebra structure 
leading to a reduced space 
${\mathfrak F}_U // {\mathfrak G}_U := {\mathfrak F}^{\mathfrak G}_U/{\mathfrak G}_U$ in which 
the pre-multisymplectic form $\Omega_L$ becomes non degenerate 
except for 
degeneracy on (the clases of) elements of ${\mathfrak F}_U$ satisfying Condition 1 for 
gauge vector fields but not Condition 2 leading to ``would be gauge degrees of freedom'' residing at the boundary; see Remark \ref{WouldBeGauge}. 
\end{remark}

\begin{remark}[Isolated systems and measuring with respect to the boundary]
\label{IsolatedSystems}
We can apply our formalism in the context of asymptotically flat General Relativity formulated a la Palatini \cite{ashtekar1990covariant}. 
The spacetime domain considered in this case 
is of the type $U = \Sigma \times [0,1]$ 
with the boundary $\partial \Sigma \times [0,1]$ being a world tube at spatial infinity (and possibly an inner boundary modelling a horizon); it is known that  diffeomorphisms induce variations such that $\iota_{jX} \Omega_L$ is horizontally exact, which implies that $X$ satisfies Condition 1 of the definition of gauge vector fields. 
However, regarding variations that do not 
vanish at infinity as gauge is inappropriate because they modify the reference frame needed to define energy, momentum and angular momentum. 
Thus, 
preserving a reference frame at the boundary that may be used as a reference for measurements is 
another motivation for Condition 2 of the definition of gauge. 

In domains of the type $U = \Sigma \times [0, 1]$ 
endowed with a foliation $\Sigma_t$, it may be desirable that 
Condition 2 is replaced by $X |_{\partial \Sigma \times [0, 1]} = 0$. 
If we use this condition all leaves $\Sigma_t$ in a foliation would be analogous to the leafs of initial and final conditions at $t= 0, 1$. 
When we are interested in an initial data formulation on a given Cauchy surface $\Sigma$, or on a Cauchy surface that has been divided into pieces, we can consider spacetime domains with a given parametrization on domains of the form $\Sigma \times [0, \epsilon]$ and take a limit. In this scenario, this is the natural version of Condition 2 to use. 
\end{remark}

\begin{remark}[``Would be gauge'' degrees of freedom at the boundary]
\label{WouldBeGauge}
Condition 2 in the definition of gauge vector fields had the main purpose of allowing for a local description of physics. 
For the sake of this discussion consider the Lie algebra of evolutionary vector fields $\hat{{\mathfrak G}}_U \subset {\mathfrak F}_U$ satisfying Condition 1 of the definition of gauge without imposing Condition 2; we could call them ``may be gauge'' vector fields over $U$.

%%
%STARTS WRONG
%In the following heuristic argument based on  
%the local product structure described in Remark \ref{GeqCl} we will  
%think of an action of the group of gauge transformations $G$ on a neighborhood of 
%$\mathcal{E}_{L, U}$ leading to the bundle $(\mathcal{E}_{L, U}/G)_\Delta$ over $U$, and we will also think of a larger group of transformations 
%$\hat{G}$ induced by the vector fields that would be gauge if we ignore Condition 2. 
%FINISHES WRONG
%%

Notice that ${\mathfrak G}_U \subset \hat{{\mathfrak G}}_U$ is a Lie ideal. The quotient, denoted by 
$(\hat{{\mathfrak G}} / {\mathfrak G})_{\partial U}$, is characterized by evolutionary vector fields in the jet over $\partial U$ which are extendible to ``would be gauge'' vector fields on the bundle over $U$. 
%The action of the pre-multisymplectic form $\Omega_L$ on $(\hat{{\mathfrak G}} / {\mathfrak G})_U$ can be defined and vanishes. 
%
Following Remark \ref{GeqCl} denote the Lie subalgebra of 
${\mathfrak F}_U$ preserving $\hat{\mathfrak G}_U$ by 
${\mathfrak F}^{\hat{\mathfrak G}}_U$ and notice that since $\hat{\mathfrak G}_U \subset {\mathfrak F}^{\hat{\mathfrak G}}_U$ is a Lie ideal 
we obtain the reduced Lie algebra 
${\mathfrak F}_U // \hat{\mathfrak G}_U := {\mathfrak F}^{\hat{\mathfrak G}}_U/\hat{\mathfrak G}_U$ in which 
the pre-multisymplectic form $\Omega_L$ becomes non degenerate in the appropriate sense. 
Clearly we also have 
${\mathfrak F}_U // \hat{\mathfrak G}_U = ({\mathfrak F}_U // {\mathfrak G}_U) // (\hat{{\mathfrak G}} / {\mathfrak G})_{\partial U}$. 
%In other words, all the remaining degeneracy of the pre-multisymplectic form in ${\mathfrak F}_U // {\mathfrak G}_U$ is removed by 
%the would be gauge degrees of freedom $(\hat{{\mathfrak G}} / {\mathfrak G})_U$ residing at the boundary. 

Notice that Condition 1 only requires 
$(\iota_{jX} \; \Omega_L)|_{{\mathfrak F}_U , \mathcal{E}_{L, U}}$ to be horizontally exact, which means that 
in general for $X \in \hat{{\mathfrak G}}_U$ 
(when $\partial \Sigma \subset \partial U$) $\iota_{jX} \; \omega_{L\Sigma}|_{{\mathfrak F}_U , \mathcal{E}_{L, U}} \neq 0$, 
and depends only on its class $[X] \in (\hat{{\mathfrak G}} / {\mathfrak G})_{\partial U}$. 
Thus, the elements of $(\hat{{\mathfrak G}} / {\mathfrak G})_{\partial U}$ should not be considered as generators of gauge transformations but as symmetry generators among boundary conditions. 

A formalism to study gauge theories in the presence of boundaries was recently put forward by Donnelly and Freidel in which boundary degrees of freedom are added to the system \cite{donnelly2016local}. 

From the perspective of our formalism the ``dynamics'' of these degrees of freedom ``added'' at the boundary is not dictated by new independent field equations. The field is bounded to be the restriction to $\partial U$ of a solution to the bulk field equation; additionally, there is a symmetry acting non trivially over those degrees of freedom generated by 
$(\hat{{\mathfrak G}} / {\mathfrak G})_{\partial U}$. 
We will mention 
in Remark \ref{ConsCharWouldBeGauge} that 
a class of vector fields in $(\hat{{\mathfrak G}} / {\mathfrak G})_{\partial U}$ 
which comes from local Lagrangian symmetries may have an associated Noether current which does not vanish. 
All these properties seem to be in agreement with \cite{donnelly2016local}. It would be interesting to have a detailed understanding of the relation between the formalism that we describe here and theirs. 
\end{remark}

\begin{remark}[Gluing spacetime domains]
Consider a domain that is constructed by gluing two subdomains $U= U_1 \cup U_2$ over a codimension one surface $\Sigma = U_1 \cap U_2$. Some gauge vector fields over $U$ are composed by a pair of a gauge vector fields over $U_1$ and a gauge vector fields over $U_2$. Notice that due to Condition 2 the given pair trivially satisfies the continuity condition at $\Sigma$, and it also trivially satisfies the linearized gluing field equation due to Condition 1. 
That is, joining gauge vector fields from $U_1$ and $U_2$ we can construct a Lie ideal of the Lie algebra of gauge vector fields on the spacetime domain $U$; 
we will write ${\mathfrak G}_{U_1} \#_\Sigma {\mathfrak G}_{U_2} := {\mathfrak G}_U^{0\Sigma} \subset {\mathfrak G}_U$. 
However, there are some gauge vector fields at $U$ that do not vanish over $\Sigma$. As mentioned in the previous remarks, these gauge vector fields 
when considered over $U_i$ were symmetry generators and after the domains are glued they become gauge vector fields. 
Thus, in order to model the Lie algebra of vector fields in the space of solutions modulo gauge 
we can calculate the quotient ${\mathfrak F}_U // {\mathfrak G}_U$ in each subdomain, glue the resulting spaces and then reduce by the gauge vector fields of $U$ that are not gauge in $U_1$ and $U_2$. 

We will write ${\mathfrak F}_{U_1} \#_\Sigma {\mathfrak F}_{U_2}$ to mean pairs of elements $V_i \in {\mathfrak F}_{U_i}$ which satisfy: (i) a zeroth order continuity condition at $\Sigma$ when evaluated on solutions (implying continuity in all the partial derivatives along $\Sigma$; which may be written as $j_\Sigma (V_1|_{\mathcal{E}_{L, U} , \Sigma}) = j_\Sigma (V_2|_{\mathcal{E}_{L, U} , \Sigma})$), 
(ii) 
the linearized gluing field equation at $\Sigma$ (\ref{LGEQ}). Notice that this condition involves directional derivatives normal to $\Sigma$, but in the presence of gauge the requirement would not be totally rigid resulting in accepting fields that are not included in ${\mathfrak F}_U$. 

${\mathfrak G}^{\hat{\Sigma}}_{U_i}$ will denote the 
subalgebra of ${\mathfrak F}_{U_i}$ whose elements satisfy Condition 1 for gauge vector fields but they do not have to vanish over $\partial U_i$ in the sense of Condition 2 because the condition may fail over $\Sigma^\circ$. 
Also we will write 
${\mathfrak G}_\Sigma = ({\mathfrak G}^{\hat{\Sigma}}_{U_1} \#_\Sigma {\mathfrak G}^{\hat{\Sigma}}_{U_2})/
({\mathfrak G}_{U_1} \#_\Sigma {\mathfrak G}_{U_2})$. 

With this notation at hand, gluing of spaces of solutions of the linearized field equation modulo gauge corresponding to adjacent domains behaves as follows  
\[
{\mathfrak F}_U // {\mathfrak G}_U = ({\mathfrak F}_U // {\mathfrak G}_U^{0\Sigma}) // ({{\mathfrak G}}_U / {\mathfrak G}_U^{0\Sigma}) = 
(({\mathfrak F}_{U_1} // {\mathfrak G}_{U_1}) \#_\Sigma ({\mathfrak F}_{U_2} // {\mathfrak G}_{U_2} )) // {\mathfrak G}_\Sigma . 
\]

This gluing procedure also seems to be in agreement with the construction of Donnelly and Freidel 
\cite{donnelly2016local}. 
We will return to the subject of gluing subdomains further ahead in the paper when we consider the algebras of observable currents associated to spacetime domains. 

In Section \ref{NestedAndGluedDomains} we discuss gluing of spaces of observable currents corresponding to adjacent domains. 
\end{remark}

%%%%%%%%%%%%%%%%%%%

\section{Observable currents}
\label{OCsSection}

%%%%%%%%%%%%%%%%%%%%%%%%%%%%

Physical observables, functions of the space of solutions modulo gauge, may be constructed by integration of currents on hypersurfaces as in formula (\ref{OCf}). 
Currents that are defined only in a certain neighborhood of the jet will be called locally defined currents; 
it turns out that considering locally defined currents is necessary to obtain the rich family of currents allowing us to 
gather all the information from systems with local degrees of freedom. 
Locally defined 
gauge invariant conserved currents are the central object of this work; in order to emphasize the use that we will give them, we will call them observable currents. 

% I erased everything regarding week OCs. If I want to retrieve it I would have to 
% go to Homero's manuscript (before I changed the name to OCs. 

In order to consider a large enough family of physical observables we will need to consider currents depending on an arbitrary number of partial derivatives of the field. This will be easily done working with differential forms in the infinite jet; the reader is reminded that 
if the current under consideration only depends on partial derivatives of the field up to order $k$ then the current could be seen as living in the $k$-jet, which is a finite dimensional manifold. We will work with integrals of the form $\int_\Sigma j \phi^\ast F$ 
(where the boundary of the hypersurface $\Sigma$, if it has one, is required to be contained in the boundary of the spacetime domain of interest $\partial \Sigma \subset \partial U$) 
and they will yield functions of the solution depending locally on the field and its partial derivatives. 

\begin{dfn}[Observable currents] 
(i) A locally defined current $F \in \mathsf{\Omega}^{n-1, 0}({\cal U} \subset J Y)$ is {\em conserved} if 
$
	\mathsf{d_h} F|_{{\cal U} \cap \mathcal{E}_{L, U}}
	= 0
$. 
(ii) It is \emph{gauge invariant} if 
$
	\mathscr{L}_{jX} F |_{{\cal U} \cap \mathcal{E}_{L, U}} 
$ 
is horizontally exact 
for every $X\in {\mathfrak G}_U$.%
\footnote{
Equivalently, one may demand that 
$j\phi^\ast \mathscr{L}_{jX} F$ be exact for any solution with $j\phi (U) \subset {\cal U}$. 
}
\\
An \emph{observable current} is a
locally defined gauge invariant conserved current in $U$. 
We will write $F \in {\rm OC}_{{\cal U}\, U}$, and ${\rm OC}_U$ will be used 
for the space of observable currents with domain in some neighborhood ${\cal U}$.
\end{dfn}

The objective of any current $F \in {\sf \Omega}^{n-1, 0}({\cal U} \subset J Y)$, in its whole existence, is to be paired with an oriented hypersurface 
$\Sigma$ so they together beget a function 
$f_\Sigma : \mathrm{Hists}_U \to \mathbb{R}$ 
through integration (defined in a certain domain) 
\begin{equation}
\label{f}
f_\Sigma [\phi] = \int_\Sigma j \phi^\ast F . 
\end{equation}
The label of the function is an oriented hypersurface $\Sigma$, 
but the conservation law obeyed by observable currents 
disregards $\Sigma$ as unimportant (except for its homology class) 
and most of the features of the function $f_\Sigma$ have origin in 
\[
	F \in {\rm OC}_U . 
\]
Notice that 
if 
$\Sigma_1 \sim \Sigma_2$ and the hypersurfaces 
are not cycles then 
$\partial \Sigma_1 = \partial \Sigma_2$.

Functions induced by observable currents 
defined by equation (\ref{f}) are gauge invariant if $\partial \Sigma \subset \partial U$. 
When we restrict these functions to act on solutions we will call them physical observables. 
\begin{dfn}[Observables from observable currents] 
The complete space spanned by 
locally defined physical observables 
\[
f_\Sigma : {\rm Sols}_U \to \mathbb R
\]
associated to a hypersurface such that 
$\partial \Sigma \subset \partial U$ 
will be denoted by ${\rm Obs}_\Sigma$. 
\end{dfn}

Notice that two observable currents $F, G \in {\rm OC}_U$ such that 
$F = G + \mathsf{d_h} \sigma$ with $\sigma|_{\mathcal{E}_{L, U}, \, \partial U}=0$ 
yield the same observable after integration on any hypersurface with 
$\partial \Sigma \subset \partial U$. It is natural to regard such pair of observable currents as equivalent; this would be the natural notion of horizontal cohomology classes in our setting, and each observable current may be thought of as a representative of its $U$-restricted horizontal cohomology class. 

\begin{remark}[Domains with a foliation]
In a domain of the type $U = \Sigma \times [0, 1]$ endowed with a foliation $\Sigma_t$ we may be interested in studying evolution of functions $f_{\Sigma_t}$ as functions of the ``time'' parameter. In this situation 
the conservation law tells us that 
$f_{\Sigma_t} - f_{\Sigma_0} = f_{\partial \Sigma \times [0, t]}$. Thus, if the boundary conditions and 
$F \in {\rm OC}_U$ are such that $f_{\partial \Sigma \times [0, t]}= 0$ the conservation law will simply state that 
the value of $f_{\Sigma_t}$ is time independent. In particular, this is expected to be the case 
for a class of observable currents of physical interest 
when the physical boundary 
$\partial \Sigma \times [0, t]$ is located ``at infinity''. However, in general the term 
$f_{\partial \Sigma \times [0, t]}$ will be relevant. 
This remark also applies to the conservation of the pre-symplectic form $\omega_{\Sigma_t}$. 

Gauge invariance is delicate because $\Sigma_t$ may have a boundary. The requirement of gauge invariance is that 
$\mathscr{L}_X F|_{\mathcal{E}_{L, U}}$ be horizontally exact, but 
since gauge vector fields satisfy 
$X|_{\partial \Sigma \times [0,1]}= 0$ the differential form $\sigma$ satisfying $\mathscr{L}_X F|_{\mathcal{E}_{L, U}} = \mathsf{d_h} \sigma|_{\mathcal{E}_{L, U}}$ vanishes at the boundary, 
$\sigma|_{\partial \Sigma \times [0,1]}= 0$. Then, in the general case $f_{\Sigma'}$ is not gauge invariant when $\Sigma'$ is not a cycle due to boundary terms, but if $\partial \Sigma' \subset \partial \Sigma \times [0,1]$ 
the boundary term breaking gauge invariance vanishes and the function is gauge invariant. Thus due to Condition 2, functions $f_{\Sigma_t}$ associated to the leaves of the foliation $\Sigma_t$ are gauge invariant. 

In Remark \ref{IsolatedSystems} we mentioned that in domains endowed with a foliation one may opt to replace Condition 2 in the definition of gauge vector fields 
by $X |_{\partial \Sigma \times [0, 1]} = 0$. 
This modification has the effect of making all leaves with $t \in [0,1]$ equivalent. 
\end{remark}

Noether's theorem stating that symmetries lead to conserved quantities is crystal clear in this framework. 
\begin{tma}[Noether]
A Lagrange symmetry is an evolutionary vector field 
$V$ satisfying 
$
	\mathscr{L}_{jV}{{ L}} 
	= 
	{\sf d_h} \sigma_L^{V}
$. 
Every Lagrange symmetry has a corresponding Noether current $N^V \in {\rm OC}_U$ 
given by 
\[
	N^V
	=
	- \iota_{jV}\Theta_L - \sigma_L^{V} . 
\]
\end{tma}
We include a proof of this classical theorem in the appendix; a more detailed presentation can be found in \cite{Sardanashvily(2016)}. 
Proving conservation of the Noether current is trivial, but 
on the other hand gauge invariance requires the use of a lemma of Takens \cite{Takens(1979)}. 

A large family of observable currents is given below. Given our previous definitions the proof of this result is simple, we state it as a theorem because 
in the context of multisymplectic formulations of classical field theory the existence of a rich family of gauge invariant conserved currents 
has been a long standing problem (see for example \cite{Forger+Romero, Helein, kanatchikov1998canonical, Kijowski, Goldschmidt(1973)}). 
We elude the obstruction 
faced by the mentioned references due to two 
special features of our approach: first of all we work in the infinite jet, which is common place in treatments based on the variational bicomplex \cite{reyes2004covariant, Vitagliano(2009)}, and second we allow for locally defined currents. 

\begin{dfn}[Symplectic product current]
Given a pair of 
vector fields in the space of solutions modeled by 
$V,W\in\mathfrak{F}_U$ their \emph{symplectic product} is the conserved current 
(defined in the intersection of the domains of definition of the perturbations)
\[
F^{VW}=\iota_{jW}\iota_{jV}\Omega_L .
\]
\end{dfn}

\begin{tma}
If the vector fields 
are ${\mathfrak G}_U$ preserving, 
$V,W\in {\mathfrak F}^{\mathfrak G}_U \subset \mathfrak{F}_U$, 
their associated symplectic product current 
is an observable current 
\[
F^{VW} \in {\rm OC}_U .
\]
\end{tma}
\begin{proof}
The following statements hold 
in the intersection of the domains of definition of $V$ and $W$. 
Conservation of $F^{VW}$ is the statement that the multisymplectic formula, described in Section \ref{Framework}, 
holds. 
Gauge invariance follows from the gauge invariance of $\Omega_L$ and the ${\mathfrak G}_U$ preserving property of $V$ and $W$. 
\end{proof}

In Section \ref{Ex} we give the elements to evaluate symplectic product observable currents in the case of the Maxwell field. 
% *** The statement below must be wrong ***
%However, the resulting observable currents are trivial due to the linearity of the field. 
An explicit example (with spacetime being one dimensional) 
showing that symplectic product observable currents are generically nontrivial is rigid body motion \cite{abraham1978foundations}, where the configuration of the system at time $t$ is given by 
$q \in SO(3)$. 
Let us denote left invariant vector fields in $SO(3)$ by $\xi \in {\mathfrak X}(SO(3))$. 
In the first order Lagrangian framework, the state of the system at time $t$ is given by $(q, \xi_q)\in TSO(3)$. 
Perturbations corresponding to generators of rotations may be parametrized by left invariant vector fields in $SO(3)$; let us denote such perturbations 
in the space of first order data 
by 
$V^{\xi} \in {\mathfrak X}(TSO(3))$. 
Consider the system at time $t=0$ at state $(q, \xi_q)\in TSO(3)$ and two perturbations of the system at that time $V^{\xi_1}$ and $V^{\xi_2}$. Evolution according to the Euler-Lagrange equation will yield $(q(t), \xi(t)_{q(t)})\in TSO(3)$; the perturbations will also evolve according to the linearized equation, and yield 
$V^{\xi_1}(t)$ and $V^{\xi_2}(t)$. 
The evaluation of the symplectic product $f^{V^{\xi_1(t)} V^{\xi_2(t)}}$ using 
the symplectic form in $TSO(3)$ induced by the Lagrangian (or equivalently by Legendre transformation of the symplectic form of 
$T^\ast SO(3)$) is 
$\omega_L (V^{\xi_1}(t), V^{\xi_2}(t))_{(q(t), \xi(t)_{q(t)})} 
= - d \theta_L (V^{\xi_1}(t), V^{\xi_2}(t))_{(q(t), \xi(t)_{q(t)})}$, where the symplectic potential is basically the angular momentum calculated in the body reference frame. 
The body angular momentum is not constant in time and the perturbations also evolve in time, but their combination in $f^{V^{\xi_1}(t) V^{\xi_2}(t)}$ is a conserved quantity. 
The Hamiltonian vector field associated to the observable shown above is the commutator of the vector fields, $[V^{\xi_1}(t), V^{\xi_2}(t)]$. 
This is a family of conserved quantities parametrized by the choice of two elements of the Lie algebra $\xi_1, \xi_2$ which encode information regarding the state of the system. 
A much larger family of observables is obtained if the perturbations of the initial conditions are allowed to depend in the state of the system at the initial time; in Lagrangian language we would say that the 
vector fields $V^{\xi_1}$ and $V^{\xi_2}$ are allowed to be {\em state dependent}. 
The same logic can be used in the case of the Yang-Mills field 
\cite{Helein:2014fta} to obtain nontrivial explicit observable currents of the symplectic product type.

\begin{remark}[Observable currents in linear field theories]
\label{LineFieldsOCs}
If we have a theory in which ${\rm Sols}_U$ is a linear subspace of ${\rm Hists}_U$ then the spaces ${\rm Sols}_U$ and ${\mathfrak F}_U$ may be identified. 
This trick, extensively used by Wald in the quantization of linear fields \cite{wald1994quantum}, leads to the following special type of 
observable currents $F^V \in {\rm OC}_U$ 
parametrized by an element 
$V\in {\mathfrak F}_U$, 
\[
F^V (j\phi) = \iota_{jW(\phi)} \iota_{jV} \Omega_L (j\phi) , 
\]
where $W(\phi)$ is an element of ${\mathfrak F}_U$ that is 
compatible with the solution $\phi \in {\rm Sols}_U$. 
%By construction if $V \in {\mathfrak F}_U^{\rm LH}$ 
%${\sf d_v}F^V = - \iota_{jV} \Omega_L$.  
We give an explicit example in Section \ref{Ex}. 
\end{remark}

%%%%%%%%%%%%%%%%%%%

\section{Locally Hamiltonian vector fields and \\
Hamiltonian observable currents}
\label{HVFs}

%%%%%%%%%%%%%%%%%%%%%%%%%%%%

%%%%%%%%%%%%%%%%%%%%%%%%%%%%%

%%%%%%%%%%%%%%%%%%%%%%%%%%%%

%
%

%%%%%%%%%%%%%%%%%%%%%%%%%%%%%%%%%%%%%%%%%%%%%%%%%%%%%%%%%%%%%%%%%%%%%%%%%%%%%%%%%%%%%

In the multisymplectic framework for field theory described in Section \ref{Framework} 
the core geometrical structure associated to a field theory is given by the structure of $J^1Y$ (and $J^2Y$), 
the field equations $\mathcal{E}_{L, U} \subset J^2Y$ and the pre-multisymplectic form $\Omega_L$. 
Thus, it would be natural to look for the locally defined structure preserving automorphisms of $J^1Y|_U$ (or $J^2Y|_U$). 
However, the objective just described turns out not to be ambitious enough 
because the space of vector fields generating structure preserving automorphisms is 
much smaller than the space vector fields on the space of solutions. The appropriate way to model perturbations of solutions in the jet 
uses locally defined {\em evolutionary vector fields} (described in Section \ref{Framework} and in the appendix) which moreover satisfy the linearized field equation, 
the requirement that 
$\mathscr{L}_{jV} {E}(L)|_{\mathcal{E}_{L, U}} = 0$ be horizontally exact. 
If we want to preserve the pre-multisymplectic structure, 
understood as the assignment of pre-symplectic structures to hypersurfaces (possibly with boundary and corners), 
then our requirements on the evolutionary vector fields should also include that 
\begin{equation}\label{llhh}
\mathscr{L}_{jV} \Omega_L  = {\sf d_h} \sigma^V
\end{equation}
when restricted to ${\mathfrak F}_U$ and evaluated in $\mathcal{E}_{L, U}$ 
for some boundary term such that \\
$(\sigma^V|_{{\mathfrak F}_U, \mathcal{E}_{L, U}})|_{\partial U} = 0$. 
Those locally defined evolutionary vector fields will be called {\em locally Hamiltonian vector fields}, and will be denoted by 
${\mathfrak F}_U^{\rm LH}\subset {\mathfrak F}_U$. 
In future arguments we will also refer to the space 
$\hat{\mathfrak F}_U^{\rm LH}\subset {\mathfrak F}_U$ composed by 
solutions of a version of equation (\ref{llhh}) in which no restriction is imposed on the boundary term over $\partial U$.

For a locally Hamiltonian vector field, a solution, and a hypersurface with $\partial \Sigma \subset \partial U$, 
it is easy to check that the pre-symplectic form induced by the solution $\omega_{L \Sigma}|_{{\mathfrak F}_U}$ satisfies 
$\mathscr{L}_{jV} \omega_{L \Sigma}  = 0$. 

There are two remarks relating locally Hamiltonian vector fields and gauge vector fields. 
First, 
if $X$ is a gauge vector field 
then it is locally Hamiltonian, 
${\mathfrak G}_U \subset {\mathfrak F}_U^{\rm LH}$. 
Second, locally Hamiltonian vector fields are ${\mathfrak G}_U$ preserving, 
${\mathfrak F}_U^{\rm LH} \subset {\mathfrak F}^{{\mathfrak G}}_U$; this is because 
preserving $\Omega_L$ implies preserving Condition 1 of Definition \ref{Gauge} for ${\mathfrak G}_U$, and Condition 2 of the definition is also preserved. 
Similarly, $\hat{\mathfrak G}_U \subset \hat{\mathfrak F}_U^{\rm LH}$.

The equation above says that, 
when restricted to ${\mathfrak F}_U, \mathcal{E}_{L, U}$, the form 
$\iota_{jV} \Omega_L$ is vertically closed (up to  horizontally exact terms vanishing at the boundary). Thus, it is natural to 
study if it can be promoted to be vertically exact (up to  horizontally exact terms whose vertical derivative vanishes at the boundary). More concretely, we look for an observable current $F$ such that 
\begin{equation}
\label{HOCs}
{\sf d_v} F = - \iota_{jV} \Omega_L + {\sf d_h}\sigma^F 
\end{equation}
when restricted to ${\mathfrak F}_U$ and evaluated in $\mathcal{E}_{L, U}$ for some boundary term such that \\
$(({\sf d_v}\sigma^F - \lambda^F)|_{{\mathfrak F}_U, \mathcal{E}_{L, U}}
)|_{\partial U} = 0$ for some form with 
${\sf d_h} \lambda^F |_{{\mathfrak F}_U, \mathcal{E}_{L, U}}= 0$.%
\footnote{
The presence of $\lambda^F$ is only due to the fact that $\sigma^V$ and $\sigma^F$ are uniquely defined only up to ${\sf d_h}$-closed terms because they are 
boundary terms. Arbitrarily setting $\lambda^F$ to zero would limit the formalism considerably; for example 
the crucial result stated in Remark \ref{EveryHVFhasOCs} and used in Theorem \ref{SeparationOfSols} may not hold. 
} 
This condition on the boundary term lets us define 
$\sigma^V = - {\sf d_v}\sigma^F + \lambda^F$ satisfying the condition over $\partial U$ required for $V$ to be a locally Hamiltonian vector field  (\ref{llhh}). Additionally, 
the condition is compatible with an equivalence relation among observable currents differing by horizontally exact terms vanishing over the boundary in the appropriate sense as discussed in the previous section. 

\begin{dfn}[Hamiltonian observable currents]\label{dfn:HOC}
An observable current $F \in {\rm OC}_U$ 
participating in equation (\ref{HOCs}) together with some $V \in {\mathfrak F}_U$ is called a 
{\em Hamiltonian observable current}. The equation implies that $V$ is locally Hamiltonian, and we say that 
$V \in {\mathfrak F}_U^{\rm LH}$ is associated to $F \in {\rm HOC}_U$. 

%An observable current $F \in {\rm OC}_U$ and a ${\mathfrak G}_U$ preserving solution of the linearized field equations $V \in {\mathfrak F}^{{\mathfrak G}}_U$ participating in equation (\ref{HOCs}) are called 
%{\em Hamiltonian observable current}, $F \in {\rm HOC}_U$, and {\em Hamiltonian vector field}, $V \in {\mathfrak F}_U^{\rm H} \subset {\mathfrak F}_U^{\rm LH}$. 

If the boundary term in equation (\ref{HOCs}) 
satisfies the stronger condition $(\sigma^F|_{{\mathfrak F}_U, \mathcal{E}_{L, U}})|_{\partial U} = 0$ the current will be called a {\em strict Hamiltonian observable current}, and we will write 
$F \in {\rm sHOC}_U \subset {\rm HOC}_U$.%
\footnote{
A proof of ${\rm sHOC}_U \subset {\rm HOC}_U$ is given as part of the prof of Theorem \ref{pro:F_VW-Hamiltonian}. 
} 
On the other hand, observable currents 
and associated vector fields that obey a version of equation (\ref{HOCs})
in which there is no condition at $\partial U$ on the boundary term belong to the spaces denoted by 
$\widehat{\rm HOC}_U$, $\hat{\mathfrak F}_U^{\rm LH}$. Notice that in the absence of boundaries all the variants of the space of Hamiltonian observable currents agree.
\footnote{
In the case of a domain with a foliation $U = \Sigma \times [0, 1]$ (or $\Sigma \times [0, \epsilon]$ if we are interested in initial data over an embedded hypersurface $\Sigma$) in which we are interested in integrating observable currents only on the leaves $\Sigma_t$ of the foliation, we mentioned that Condition 2 for gauge vector fields could be traded for the vanishing of the field over $\partial \Sigma \times [0, 1]$; accordingly, the condition on the boundary terms of equation (\ref{llhh}) defining locally Hamiltonian vector fields 
should be 
$(\sigma^V|_{{\mathfrak F}_U, \mathcal{E}_{L, U}})|_{\partial \Sigma \times [0, 1]} = 0$, the condition on the boundary terms of equation (\ref{HOCs}) 
concerning this definition  
should be $(({\sf d_v}\sigma^F - \lambda^F)|_{{\mathfrak F}_U, \mathcal{E}_{L, U}})|_{\partial \Sigma \times [0, 1]} = 0$ 
for some form $\lambda^F$ which is horizontally closed in the appropriate sense, 
and similarly the condition for strict Hamiltonian observable currents should be properly adjusted.} 
\end{dfn}

Let us discuss first the version of equation (\ref{HOCs}) in which no conditions are imposed on the boundary term over $\partial U$. 
Consider two elements of $\hat{\mathfrak F}_U^{\rm LH}$ compatible with the same observable current $F \in \widehat{\rm HOC}_U$ in the sense that 
$({\sf d_v} F + \iota_{jV_i} \Omega_L - {\sf d_h}\sigma_i^{F})|_{{\mathfrak F}_U , \mathcal{E}_{L, U}} = 0$. 
Then $(\iota_{j(V_1 -V_2)} \Omega_L - {\sf d_h}(\sigma_1^{F}- \sigma_2^{F}))_{{\mathfrak F}_U, \mathcal{E}_{L, U}} = 0$ 
proving that the vector fields differ by an element of $\hat{\mathfrak G}_U$. Thus, 
to each element of $\widehat{\rm HOC}_U$ corresponds a unique element of $\hat{\mathfrak F}_U^{\rm LH} / \hat{\mathfrak G}_U$, and the currents that are mapped to the zero element 
are those that are ${\sf d_v}$-constant up to ${\sf d_h}$-exact terms; in other words, these observable currents are pure boundary terms up to a field independent constant.

Now let us consider locally Hamiltonian vector fields ${\mathfrak F}_U^{\rm LH}\subset \hat{\mathfrak F}_U^{\rm LH}$ and their gauge equivalence classes. 
We know that ${\mathfrak F}_U^{\rm LH} \cap \hat{\mathfrak G}_U$ 
contains those locally Hamiltonian vector fields which are compatible with the zero observable current ${\mathfrak F}_U^{{\rm LH},0}$; they are vector fields in $\hat{\mathfrak G}_U$ whose associated boundary terms satisfy extra conditions over $\partial U$. It is simple to see that 
${\mathfrak F}_U^{{\rm LH},0} \subset {\mathfrak F}_U^{\rm LH}$ is a sub Lie algebra; however, it is not an ideal. 
On the other hand, ${\mathfrak G}_U$ is an ideal of ${\mathfrak F}_U^{\rm LH}$; 
we consider an scenario in which 
${\mathfrak G}_U$ is the maximal ideal contained in ${\mathfrak F}_U^{\rm LH} \cap \hat{\mathfrak G}_U$. 
That this scenario takes place should be tested in each field theory of interest. 
Under this assumption, there is an assignment of a unique element of ${\mathfrak F}_U^{\rm LH} / {\mathfrak G}_U$ 
to each element of ${\rm HOC}_U$. 
The currents that are mapped to the zero element (i.e. currents associated to gauge vector fields) 
are those that are ${\sf d_v}$-constant up to a boundary term 
(i.e. 
$({\sf d_v} F - {\sf d_h}\sigma^{F})|_{{\mathfrak F}_U , \mathcal{E}_{L, U}} = 0$)
such that 
$(({\sf d_v}\sigma^F - \lambda^F)|_{{\mathfrak F}_U, \mathcal{E}_{L, U}}
)|_{\partial U} = 0$
for some form with 
${\sf d_h} \lambda^F |_{{\mathfrak F}_U, \mathcal{E}_{L, U}}= 0$; 
in other words, these observable currents are field independent up to 
a pure boundary term 
which when evaluated over $\partial U$ yields a field independent constant.

%
%%Old 
%Now let us consider locally Hamiltonian vector fields ${\mathfrak F}_U^{\rm LH}\subset \hat{\mathfrak F}_U^{\rm LH}$ and their gauge equivalence classes. 
%We know that ${\mathfrak F}_U^{\rm LH} \cap \hat{\mathfrak G}_U \supset {\mathfrak G}_U$ and we consider reasonable to conjecture that ${\mathfrak G}_U$ is the maximal ideal in this intersection. 
%Under this assumption, there is an assignment to each element of ${\rm HOC}_U$ to a unique element of ${\mathfrak F}_U^{\rm LH} / {\mathfrak G}_U$, and the currents that are mapped to the zero element (i.e. currents associated to gauge vector fields) 
%are those that are ${\sf d_v}$-constant up to ${\sf d_h}$-exact terms with 
%%$({\sf d_v}\sigma^F|_{{\mathfrak F}_U, \mathcal{E}_{L, U}})|_{\partial U} = 0$; 
%$(\sigma^F|_{{\mathfrak F}_U, \mathcal{E}_{L, U}})|_{\partial U} = 0$; 
%%in other words, these observable currents are pure boundary terms 
%%which are field independent over $\partial U$ 
%%up to a field independent constant. 
%in other words, these observable currents consisting of a pure boundary term vanishing over $\partial U$ 
%up to a field independent constant. 
% %

\begin{remark}[Every locally Hamiltonian vector field has corresponding observable currents]
\label{EveryHVFhasOCs}
The obstruction for the existence of a Hamiltonian observable current associated to a given locally Hamiltonian vector field is the non triviality of 
the vertical cohomology group $H_{\sf d_v}^{n-1, 1}(J Y|_U)$, 
where furthermore there is an equivalence relation among vertically closed forms differing by horizontally exact terms ${\sf d_h}\sigma^V$ such that 
$(\sigma^V|_{{\mathfrak F}_U, \mathcal{E}_{L, U}})|_{\partial U} = 0$, 
and a corresponding equivalence relation among vertically exact forms that differ by horizontally exact terms ${\sf d_h}\sigma^F$ 
such that 
$(({\sf d_v}\sigma^F - \lambda^F)|_{{\mathfrak F}_U, \mathcal{E}_{L, U}}
)|_{\partial U} = 0$ for some form with 
${\sf d_h} \lambda^F |_{{\mathfrak F}_U, \mathcal{E}_{L, U}}= 0$. 
However, since we allow 
observable currents that are defined only on neighborhoods of $j \phi (U)$ 
the mentioned cohomology group is trivial. Thus,  
there is a Hamiltonian observable current for any given 
locally Hamiltonian vector field in a properly adjusted domain of definition. 
For any $j \phi (U)$ in the domain of definition of $V \in {\mathfrak F}_U^{\rm LH}$ there is a neighborhood containing it which is the domain of definition of an observable current $F \in {\rm OC}_U$ satisfying formula (\ref{HOCs}). 
This fact indicates that there are plenty of observable currents, and will be a key ingredient for proving that the physical observables calculable from observable currents can distinguish gauge inequivalent solutions (see Section \ref{OCsSeparate}). 

This feature of Definition \ref{dfn:HOC} is the reason why in this work we mention locally Hamiltonian vector fields and never Hamiltonian vector fields. It is appropriate to say that in a framework allowing for locally defined observable currents all locally Hamiltonian vector fields are automatically Hamiltonian. 
\end{remark}

\begin{remark}[Further support for the definition of gauge vector fields]
\label{FurtherSup}
Notice that every conserved current satisfying equation (\ref{HOCs}) 
is gauge invariant; Condition 1 and Condition 2 are essential for this. 
Thus, regarding observable currents as generators of multisymplectomorphisms gives further support for 
our definition of gauge vector fields. 
\end{remark}

\begin{remark}[Hamiltonian observables]
Equation (\ref{HOCs}) induces on ${\rm Obs}_\Sigma$ 
the all important equation of symplectic geometry; given any vector field in the space of solutions $w$ we have
\[
\iota_w df_\Sigma = - \iota_w \iota_v \omega_{L \Sigma} - 
\int_{\partial \Sigma} j\phi^\ast \iota_{j W } \sigma^F 
.
\] 
The resulting space of Hamiltonian observables is 
the complete subspace 
denoted by \\
${\rm HObs}_\Sigma \subset {\rm Obs}_\Sigma$. 

Notice that the familiar formula is induced in the absence of a boundary, but 
in the the general case $df_\Sigma$ acquires 
an extra term. 
%which is field independent due to the condition 
%$(({\sf d_v}\sigma^F)|_{{\mathfrak F}_U, \mathcal{E}_{L, U}}
%)|_{\partial U} = 0$. 

Strict Hamiltonian observable currents ${\rm sHOC}_U \subset {\rm HOC}_U$ induce strict Hamiltonian observables 
${\rm sHObs}_\Sigma \subset {\rm HObs}_\Sigma \subset {\rm Obs}_\Sigma$ which obey the equation 
$df_\Sigma = - \iota_v \omega_{L \Sigma}$. 
\end{remark}

Of course the first examples of Hamiltonian observable currents 
are Noether currents. 
\begin{tma}\label{NoetherThm}
A Noether current 
$
	N^V
	=
	- \iota_{jV}\Theta_L - \sigma_L^{V} 
$ satisfies the equation 
\[
({\sf d_v} N^V + \iota_{jV} \Omega_L - {\sf d_h}\sigma_N^V)|_{{\mathfrak F}_U, \mathcal{E}_{L, U}} = 0 . 
\]
Thus $N^V \in \widehat{\rm HOC}_U$. 
If 
$(({\sf d_v}\sigma_N^V - \lambda^F)|_{{\mathfrak F}_U, \mathcal{E}_{L, U}}
)|_{\partial U} = 0$, for some form $\lambda^F$ with 
${\sf d_h} \lambda^F |_{{\mathfrak F}_U, \mathcal{E}_{L, U}}= 0$, 
then it is is a Hamiltonian observable current according to our definition 
$N^V \in {\rm HOC}_U$ with $V$ as its Hamiltonian vector field. 
\end{tma}
A proof for the existence of such a horizontally exact term requires the use of Takens' lemma \cite{Takens(1979)}, and it is stated in the appendix; for details see \cite{deligne1999classical}. 
Notice that it is not a priory clear if 
the boundary term satisfies the required condition over $\partial U$; 
in general Noether currents are not Hamiltonian Observable currents when there are boundaries.

\begin{remark}[Conserved charges associated to ``would be gauge'' symmetries]
\label{ConsCharWouldBeGauge}
A result of Wald and Lee 
\cite{lee1990local, Vitagliano(2009)} says that 
an element $X$ of a family of Lagrange symmetries depending on parameters with possible arbitrary local variation satisfies Condition 1 of the definition of gauge vector fields \cite{lee1990local, Vitagliano(2009)}, 
and it has 
a corresponding conserved Noether current $N^X$. 
If $X$ also satisfies 
the locality condition in the definition of gauge vector fields 
requiring that $jX |_{\mathcal{E}_{L, U} , \partial U} = 0$, 
then 
$(\sigma_L^{X}|_{{\mathfrak F}_U, \mathcal{E}_{L, U}})|_{\partial U} = 0$
which implies that 
$N^X \in {\rm HOC}_U$ and that 
$n^X_\Sigma (\phi) = \int_\Sigma j\phi^\ast N^X = 0$ 
for any hypersurface with 
with $\partial \Sigma \subset \partial U$. 
%In this case the differential of $n^X_\Sigma$ must also vanish, which means that 
%$\sigma_N^X|_{\partial U}=0$ \cite{Olver, Zukerman, Wald+Lee}. 

Now consider one of these generators of local Lagrangian symmetries $X$ which does not vanishing over $\partial U$, a ``would be gauge'' vector field. 
The Noether charge 
$n^X_\Sigma (\phi) = \int_\Sigma j\phi^\ast N^X$ would vanish if 
$\Sigma$ is a cycle; thus, the current must be horizontally exact 
$N^X = {\sf d_h}\nu_X$. In our case we have 
\[
n^X_\Sigma (\phi) = \int_{\partial \Sigma} j\phi^\ast \nu_{X} , 
\] 
which would not vanish in general. 
Moreover, since any hypersurface $\Sigma'$ homologous with $\Sigma$ has the same boundary our ability to move the hypersurface to a region where 
the vector field vanishes (as used in the argument in the absence of boundaries) is crucially diminished, and the boundary integral in general does not vanish. 
This is consistent with the study of boundary integrals by Regge and Teitelboim \cite{regge1974role}. 

In general the boundary term $\sigma_N^X$ 
participating in equation (\ref{HOCs}) for 
a Noether current $N^X$ of this type will not follow the condition over $\partial U$ making the current not Hamiltonian according to our definition. 
One aspect of this is that two different vector fields $X, X'$ both satisfying equation (\ref{HOCs}) together with $N^X$ differ by an element of 
$\hat{\mathfrak G}_U$ which is a larger set than ${\mathfrak G}_U$. 

%One reason that we excluded this type of currents is that they may have two associated vector fields $X, X'$ both satisfying equation (\ref{HOCs}) 
%even when $X-X' \notin {\mathfrak G}_U$; another reason is that the 
%flow of these vector fields does not preserve the pre-symplectic forms $\omega_{L \Sigma}$ associated to hypersurfaces with $\partial \Sigma \subset \partial U$. 

The vertical differential of the Noether current 
associated to a ``would be gauge'' vector field (one which does not satisfy the locality condition) 
is a pure boundary term. The differential of the corresponding charge is 
\[
d n^X_\Sigma (\phi) = -\int_{\partial \Sigma} j^1\phi^\ast {\sf d_v}\nu^{X} 
= \int_{\partial \Sigma} j^1\phi^\ast (\sigma_N^X - \rho^X) . 
\] 
In abelian Chern-Simons theory over a bounded domain 
Noether charges associated to would be gauge vector fields do not vanish. 
The interested reader can perform the calculation 
following the notation presented in the example of Section \ref{Ex} 
and arrive to the results recently presented in \cite{Geiller:2017xad} expressed in our language. 
\end{remark}

\begin{tma}\label{pro:F_VW-Hamiltonian}
Let $F^{VW}$ be a symplectic product observable current associated to two locally Hamiltonian vector fields $V,W \in {\mathfrak F}^{\rm LH}_U$. Then $F^{VW} \in {\rm sHOC}_U \subset {\rm HOC}_U$ with associated locally Hamiltonian vector field $[V,W] \in  \mathfrak{F}^{\rm LH}_U$ 
\[
	\mathsf{d_v}F^{VW}=-\iota_{j[V,W]}\Omega_L + \mathsf{d_h} \sigma^{VW}
\]
where 
\[
	\sigma^{VW}
	=
	\iota_{jW}\sigma^V - \iota_{jV}\sigma^W .
\]
\end{tma}
\begin{proof}
A short calculation yields 
$\mathsf{d_v}F^{VW}=-\iota_{[jV,jW]}\Omega_L + \iota_{jV} \mathscr{L}_{jW}\Omega_L - \iota_{jW} \mathscr{L}_{jV}\Omega_L$. 
The proof is completed noticing that 
$(\mathscr{L}_{jV}\Omega_L - \mathsf{d_h} \sigma^V)|_{{\mathfrak F}_U, \mathcal{E}_{L, U}} = 0$ and 
$(\mathscr{L}_{jW}\Omega_L - \mathsf{d_h} \sigma^W)|_{{\mathfrak F}_U, \mathcal{E}_{L, U}} = 0$, and using 
$[jV,jW] = j[V,W]$ (for a brief explanation see the appendix). 

We need to prove that the boundary term satisfies the condition over $\partial U$ demanded by the definition. 
Notice first that $\sigma^{VW}$ satisfies the condition over $\partial U$ that makes the observable current a strict Hamiltonian observable current, 
$F^{VW}\in {\rm sHOC}_U$. Below we will prove that ${\rm sHOC}_U \subset {\rm HOC}_U$. 

Consider 
$({\sf d_v} F + \iota_{jV} \Omega_L - {\sf d_h}\sigma^{F})|_{{\mathfrak F}_U , \mathcal{E}_{L, U}} = 0$ with 
$\sigma^{F}|_{{\mathfrak F}_U , \mathcal{E}_{L, U}})|_{\partial U} = 0$, and an arbitrary pair of vector fields $Y, Z \in {\mathfrak F}_U$. 
We calculate ${\sf d_v} \sigma^{F}$ and 
proceed eliminating terms due to the behavior of the boundary term over $\partial U$ 
\begin{eqnarray*}
\iota_{jY} \iota_{jZ} {\sf d_v} \sigma^{F}|_{\mathcal{E}_{L, U}} 
&=& \iota_{jY}( \mathscr{L}_{jZ} - 
{\sf d_v} \iota_{jZ} ) \sigma^{F} |_{\mathcal{E}_{L, U}} = 
- \iota_{jY} {\sf d_v} \iota_{jZ} \sigma^{F} |_{\mathcal{E}_{L, U}} = - \mathscr{L}_{jY}( \iota_{jZ} \sigma^{F} )|_{\mathcal{E}_{L, U}}  \\
&=&  
- ( \iota_{j[Y, Z]} + \iota_{jZ} \mathscr{L}_{jY} ) \sigma^{F} |_{\mathcal{E}_{L, U}} = 0. 
\end{eqnarray*}
\end{proof}

Apart from describing a property of an important family of observable currents, the previous result has the following corollary. 
\begin{cor}\label{corollary1}
${\mathfrak F}_U^{\rm LH} 
\subset {\mathfrak F}_U^{\mathfrak G}$ 
is a Lie subalgebra, 
and $\left[  {\mathfrak F}_U^{\rm LH}, {\mathfrak F}_U^{\rm LH}\right]
\subset 	{\mathfrak F}_U^{\rm LH}$  
is a Lie ideal corresponding to symplectic product observable currents. 
Additionally, this structure is compatible with reduction by gauge vector fields 
producing the natural inclusions 
\[
\left[  {\mathfrak F}_U^{\rm LH}, {\mathfrak F}_U^{\rm LH}\right]/{\mathfrak G}_U \to 
{\mathfrak F}_U^{\rm LH}/{\mathfrak G}_U \to 
{\mathfrak F}_U//{\mathfrak G}_U:= {\mathfrak F}_U^{\mathfrak G}/{\mathfrak G}_U . 
\]
\end{cor}

In symplectic geometry every function of phase space has an associated Hamiltonian vector field. 
However, 
from the families of examples given above (Noether and symplectic product observable currents) 
we see that 
for field theories over confined spacetime domains 
not all observable currents are Hamiltonian observable currents. 
However, if 
$U = \Sigma \times [0, 1]$ we can use the weaker locality requirement for gauge vector fields 
(requiring them to vanish on $\partial \Sigma \times [0, 1]$
instead of vanishing in the whole $\partial U$) 
and end up with definitions according to which, in the case of $\Sigma$ being a cycle, 
all the families of observables currents exhibited here are Hamiltonian. 
Moreover, 
section \ref{OCsSeparate} proves that, 
in the presence of a Cauchy surface with no boundary, 
Hamiltonian observable currents are capable of distinguishing gauge inequivalent solutions. 
Then any observable induced by an observable current maybe approximated by observables 
calculable from Hamiltonian observable currents, which appealing to completeness implies that 
every observable current is Hamiltonian. 
Thus, 
in the presence of a Cauchy surface with no boundary every observable current is Hamiltonian; 
the corresponding statement in a context closely related to ours was rigorously proven by Vitagliano in \cite{Vitagliano(2009)}.

%%%%%%%%%%%%%%%%%%%%%%%%%%%%%

\section{A bracket for observable currents and \\
the Poisson algebra of local observables}
\label{BracketSection}
%%%%%%%%%%%%%%%%%%%%%%%%%%%%

%
%

%%%%%%%%%%%%%%%%%%%%%%%%%%%%%%%%%%%%%%%%%%%%%%%%%%%%%%%%%%%%%%%%%%%%%%%%%%%%%%%%%%%%%

Given two Hamiltonian vector fields $V, W$, with associated Hamiltonian observable currents $F, G \in {\rm HOC}_U$, Corollary \ref{corollary1} tells us that 
their Lie product $[V, W]$ is 
another Hamiltonian vector field. We would like to find a Hamiltonian observable current associated to $[ V, W]$. 
It would be even nicer if the resulting observable current could be calculated only from $F$ and $G$ and the assignment made the vector space of 
Hamiltonian observable currents ${\rm HOC}_U$ into a Lie algebra extending the Lie algebra of locally Hamiltonian vector fields ${\mathfrak F}_U^{\rm LH}$. 
Below we will show several different Hamiltonian observable currents which have 
$[ V, W]$ as their Hamiltonian vector field; they differ by horizontally exact terms vanishing over $\partial U$. 
Thus, when these different candidates 
are integrated over hypersurfaces with $\partial \Sigma \subset \partial U$ 
they all coincide. 
In this way ${\rm HObs}_\Sigma$ acquires the structure of a Lie algebra. 
Furthermore, if there is a Cauchy surface $\Sigma$ without boundary 
we will see below that ${\rm HObs}_\Sigma$ becomes a Poisson  algebra. 

Consider any two Hamiltonian observable currents $F, G \in {\rm HOC}_U$ with choices of locally Hamiltonian vector fields $V,W \in\mathfrak{F}^{\rm LH}_U,$ respectively. 
We have already shown (see Proposition \ref{pro:F_VW-Hamiltonian}) that $F^{VW} \in {\rm HOC}_U$ with Hamiltonian vector field $[V, W]$. This gives us a natural definition of a bracket among Hamiltonian observable currents; however, different choices of Hamiltonian vector fields for $F$ and $G$ lead to results that differ by a boundary term. 
Here is the formal statement. 
\begin{dfn}[Bracket for observable currents] Let $F, G \in {\rm HOC}_U$ with choices of locally Hamiltonian vector fields $V,W \in \mathfrak{F}^{\rm LH}_U$ respectively. The bracket 
\begin{equation}\label{Poisson-Lie brackets}
	\{F^V,G^W\} = \iota_{jW} \iota_{jV}	\Omega_L 
\end{equation}
defines a Hamiltonian observable current $\{F^V,G^W\} \in {\rm HOC}_U$. 
\end{dfn}

\begin{remark}[Dependence on the choice of Hamiltonian vector fields]
Consider $V_1, V_2$ locally Hamiltonian vector fields for $F \in {\rm HOC}_U$. Thus, 
\[
\{F^{V_1},G^W\} - \{F^{V_2},G^W\}
= \iota_{jW} \iota_{jV_1-jV_2}	\Omega_L = - {\sf d_h} \iota_{jW} \rho^{V_1-V_2}  
\]
since, as shown just after Definition \ref{dfn:HOC}, $V_1-V_2 \in \hat{\mathfrak G}_U$. Thus, different choices of locally Hamiltonian vector fields lead to observable currents differing by a boundary term that is not simple in principle. 
If the scenario described 
after Definition \ref{dfn:HOC} takes place in the field theory of interest 
$V_1-V_2 \in {\mathfrak G}_U$ and the boundary term vanishes over $\partial U$ when evaluated on $\mathcal{E}_{L, U}$. 
The induced bracket on ${\rm sHObs}_\Sigma$ would be independent of any choice. 

In the absence of a boundary, or if use the weaker condition on the boundary terms appropriate for domains with a foliation, 
the boundary terms would not be relevant for ${\rm HObs}_\Sigma$ and the induced bracket 
is independent of the choice of Hamiltonian vector field, and we can write $\{F,G\}$. 
\end{remark}

The following result was proven in the previous section. 
\begin{lma}\label{lma:above-1}
Let $F,G  \in  {\rm HOC}_U$ be observable currents with locally Hamiltonian vector fields  $V,W \in {\mathfrak F}^{\rm LH}_U$ respectively. 
Then 
\[
{\sf d_v} \{F^V,G^W\} = - \iota_{j[V,W]} \Omega_L + {\sf d_h}\sigma^{VW}, \quad \mbox{with } \sigma^{VW} = \iota_{jW} \sigma^V - \iota_{jV} \sigma^W 
\]
where 
$\sigma^{VW}$ satisfies the conditions at $\partial U$ necessary to make $\{F^V,G^W\}$ a strict Hamiltonian observable current with $[V, W]$ as associated Hamiltonian vector field. 
\end{lma}

Other Hamiltonian observable currents with $[V, W]$ as Hamiltonian vector field are 
$\mathscr{L}_{jV} G$ and $-\mathscr{L}_{jW} F$, which are geometrically interesting since they associate observable currents to Lie derivatives along 
vector fields in the jet. However, they have the disadvantage of not being skew symmetric, but it is also possible to skew symmetrize them. Here is the relation between the mentioned Hamiltonian observable currents. 
\begin{eqnarray*}
	\{F^V,G^W\} &=& \mathscr{L}_{jV} G + {\sf d_h} \iota_{jV} \sigma^G = - \mathscr{L}_{jW} F - {\sf d_h} \iota_{jV} \sigma^F \\
	&=& \frac{1}{2} \left( \mathscr{L}_{jV} G - \mathscr{L}_{jW} F \right) + 
	\frac{1}{2} {\sf d_h} \left( \iota_{jV} \sigma^G - \iota_{jV} \sigma^F \right)  
\end{eqnarray*}
which holds 
when evaluated in $\mathcal{E}_{L, U}$. 

It is clear that our bracket is bilinear and skew symmetric. However, it does not satisfy a Jacobi relation. On the other hand, 
it is a straight forward calculation to verify that the Lie derivative bracket 
$\{F^V,G\}_l = \mathscr{L}_V G$, which is not skew symmetric, satisfies a Jacobi identity 
\begin{eqnarray*}
	\left\{F^{V_1}_1,\{F^{V_2}_2,F_3\}_l\right\}_l &=& 
	\mathscr{L}_{V_1} \mathscr{L}_{V_2} F_3 = 
	\mathscr{L}_{[V_1, V_2]} F_3 + 
	\mathscr{L}_{V_2} \mathscr{L}_{V_1} F_3 \\
	&=& 
	\left\{\{F^{V_1}_1,F_2\}^{[V_1, V_2]}_l,F_3\right\}_l + 
	\left\{F^{V_2}_2,\{F^{V_1}_1,F_3\}_l\right\}_l 
. 
\end{eqnarray*}
Now, repeated use of the identity 
$(\{F^V,G\}_l - \{F^V,G^W\} + {\sf d_h} \iota_{jV} \sigma^G)|_{\mathcal{E}_{L, U}}= 0$ 
lets us see that our bracket, when evaluated in $\mathcal{E}_{L, U}$, 
is subject to a Jacobi relation that is modified by a horizontally exact term 
\[
	\left\{F^{V_1}_1,\{F^{V_2}_2,F^{V_3}_3\}\right\} + {\sf d_h} J
	=
	\left\{\{F^{V_1}_1,F^{V_2}_2\},F^{V_3}_3\right\} + 
	\left\{F^{V_2}_2,\{F^{V_1}_1,F^{V_3}_3\}\right\} , 
\]
with $J = -\iota_{jV_1}\sigma^{V_2 V_3} + \iota_{jV_2}\sigma^{V_1 V_3} + 
( \iota_{jV_1}\mathscr{L}_{jV_2} - \iota_{jV_2}\mathscr{L}_{jV_1} - \iota_{j[V_1, V_2]}) \sigma^{F_3}$. Notice that, 
since $V_1, V_2, V_3 \in {\mathfrak F}_U$ and 
$(\sigma^{V_i}|_{{\mathfrak F}_U, \mathcal{E}_{L, U}})|_{\partial U} = 0$ 
the first two terms 
in $J$ vanish over $\partial U$; 
however, the third term does not necessarily vanish over $\partial U$.

\begin{remark}[Lie $n$-algebra of observable currents]
The structure in 
${\rm HOC}_U$ given by the brackets defined above 
fits into the general structure described by Rogers \cite{rogers2012algebras} 
for multisymplectic field theory. Baez and Rogers studied the case of 
the classical bosonic string with particular detail \cite{baez2010categorified}. 
Our bracket $\{F^V,G^W\}$ corresponds to the hemibracket, and $\{F^V,G\}_l$ corresponds to semibracket in their notation. 

In \cite{barnich1998sh} Barnich et al use the variational bicomplex to develop a framework appropriate for spacetime-localized observables. For us it would be of great interest to understand the relation between their work and ours. 

There is further work \cite{fiorenza2014l_} with the motivation of studying algebraic properties of Noether currents in multisymplectic field theory. 
\end{remark}

If we are working on a domain $U = \Sigma \times [0, 1]$ 
with a foliation we can use the weak condition on the boundary terms $\sigma^{V_i}$ and $\sigma^{F_i}$ that take place over 
$(\partial \Sigma) \times [0, 1]$ which has the advantage that if the leaves $\Sigma_t$ do not have a boundary the conditions 
on boundary terms become trivially satisfied. In that case the calculation given above showing that the Jacobi relation is modified by 
a horizontally exact term ${\sf d_h} J$ implies that the bracket will induce a Lie algebra structure in ${\rm HObs}_\Sigma$. 

We also defined a restricted class of Hamiltonian observable currents, ${\rm sHOC}_U$, for which the boundary terms satisfy 
$(\sigma^{F_i}|_{{\mathfrak F}_U, \mathcal{E}_{L, U}})|_{\partial U} = 0$. 
Lemma \ref{lma:above-1} implies that ${\rm sHOC}_U$ is a subalgebra of ${\rm HOC}_U$. 
For this class of Hamiltonian observable currents the modification of the Jacobi identity is by a horizontally exact term which 
vanishes over $\partial U$ (when evaluating on $\mathcal{E}_{L, U}$). Thus, for any hypersurface with $\partial \Sigma \subset \partial U$ 
the bracket given above induces a Lie algebra structure in ${\rm sHObs}_\Sigma$. 
Before stating the result formally, we recall that as any space of functions ${\rm HObs}_\Sigma$ is endowed with the spacetime non-local product of pointwise evaluation 
$(f\cdot g)_\Sigma [\phi] = f_\Sigma [\phi] g_\Sigma [\phi]$. 
In Remark \ref{GenObsFromOCs} we comment on 
the nontrivial issue of whether any product observable is realizable as the integral of a current. 
If $\Sigma$ has no boundary the degeneracy of $\omega_{L\Sigma}$ is completely characterized by 
${\mathfrak G}_U = \hat{\mathfrak G}_U \subset {\mathfrak F}_U$; on the other hand, since observables in ${\rm Obs}_\Sigma$ are required to be 
gauge invariant it is reasonable to expect that ${\rm Obs}_\Sigma = {\rm HObs}_\Sigma$. 
Furthermore, if $\Sigma$ is a Cauchy surface then it is reasonable to assume that all observables are generated by integrals of currents over $\Sigma$. 

\begin{pro}\label{BracketObs}
\begin{enumerate}
\item
If $\Sigma$ is a hypersurface with no boundary the bracket induced by the equation 
\begin{equation}
\label{BracketObsEq}
[f_\Sigma , g_\Sigma ]_\Sigma = \int_\Sigma j \phi^\ast \{ F^V , G^W\} 
\end{equation}
gives 
${\rm HObs}_\Sigma$ the structure of a Lie algebra. 
\item
If $\partial \Sigma \subset \partial U$ the same bracket makes ${\rm sHObs}_\Sigma$ a Lie algebra. 
\item
If $\Sigma$ is a Cauchy surface without boundary, 
${\rm HObs}_\Sigma$ is closed under 
the product of pointwise evaluation and it acquires 
the structure of a Poisson algebra. 
\end{enumerate}
\end{pro}
\begin{proof}
For a hypersurface with no boundary 
the bracket $[\cdot , \cdot ]_\Sigma$ in ${\rm HObs}_\Sigma$ inherits bilinearity and skew symmetry from the bracket $\{ \cdot , \cdot \}$ in 
${\rm HOC}_U$. 
Jacobi's identity holds because, 
after integration on a hypersurface with no boundary, 
the modification of the pure boundary term 
${\sf d_h}J$ modifying Jacobi's relation is irrelevant. 

In the case with $\partial \Sigma \subset \partial U$ and observables in ${\rm sHObs}_\Sigma$
the boundary term ${\sf d_h}J$ vanishes over $\partial U$ and again after integration it vanishes.

If $\Sigma$ is a Cauchy surface without boundary,
the assumption guarantees that bracket observables 
and the product observables are again in 
${\rm HObs}_\Sigma$; 
for the product observables this is shown in Remark \ref{GenObsFromOCs}. 
In addition, 
the Leibnitz's rule is satisfied because 
for any $f_\Sigma \in {\rm HObs}_\Sigma$ the bracket induces a derivative operator 
$[f_\Sigma , g_\Sigma ]_\Sigma = \int_\Sigma j^1\phi^\ast \mathscr{L}_V G$. 
\end{proof}

These results allow us to refine Theorem \ref{NoetherThm} about Noether currents; the following result is a corollary of that theorem and Lemma 
\ref{lma:above-1}. 

\begin{cor}[Algebra of Noether currents]
A Lie algebra of Lagrange symmetries $\mathscr{S}_L \subset {\mathfrak F}_U$ 
induces a vector space of observable currents 
${\rm O}_{\mathscr{S}_L} \subset {\rm HOC}_U$ which is compatible with the brackets in the sense that 
given $V, W \in \mathscr{S}_L$ we have 
\[
\{ N^V , N^W \} = N^{[V, W]} + {\sf d_h} \sigma_N^{VW} . 
\]
with boundary term $\sigma_N^{VW}= \iota_W \sigma_N^V - \iota_V \sigma_N^W$. 
%%%%%
Notice that 
the boundary term will satisfy 
$(\sigma_N^{VW}|_{\mathcal{E}_{L, U}})|_{\partial U} = 0$
if the symmetry algebra obeys the locality condition \\
$V \in \mathscr{S}_L \implies (\mathscr{L}_V \Omega_L - {\sf d_h} \sigma^V)|_{{\mathfrak F}_U, \mathcal{E}_{L, U}})|_{\partial U} = 0$ with 
$(\sigma^V|_{{\mathfrak F}_U, \mathcal{E}_{L, U}})|_{\partial U} = 0$. 

If 
the locality condition written above is satisfied%
\footnote{
Notice that the appropriate version of the condition is trivially satisfied if 
we are working on a domain $U = \Sigma \times [0, 1]$ 
with a foliation and the leaves $\Sigma_t$ do not have a boundary.} 
by the symmetry algebra then 
for any $V \in \mathscr{S}_L$ we have 
$n^V_\Sigma = \int_\Sigma j^1\phi^\ast N^V \in {\rm HObs}_\Sigma$. 
Moreover, 
${\rm HObs}_\Sigma$ 
acquires the structure of a Lie algebra and 
the correspondence is a Lie algebra morphism 
\[
\mathscr{S}_L \to {\rm HObs}_\Sigma \quad . 
\]
\end{cor}

%*** REFEREE RECOMMENDS TO STAY OUT OF TECHNICAL REMARKS ABOUT LIE N-ALGS ***
%
%A general framework to study the algebraic properties of Noether currents in multisymplectic field theory extending the work of Rogers \cite{rogers2012algebras} is developed in \cite{fiorenza2014l_}. 
%With the aim of understanding the association of currents to symmetries in a finer way a framework for homotopy moment maps was introduced in 
%\cite{callies2016homotopy}. 

Symplectic product observable currents are a large family of observable currents. In the special case when the vector fields ${V,W}$ are locally Hamiltonian we gave explicit formula for the Hamiltonian vector field associated to $F^{VW}$. 
The following result about the algebra of symplectic product currents is a trivial consequence of the definitions, and complements Corollary \ref{corollary1}. 
\begin{pro}[Algebra of symplectic product currents]
Let $F^{V_1W_1}$ and $F^{V_2W_2}$ be symplectic product observable currents associated to the locally Hamiltonian vector fields $V_1,W_1; V_2, W_2 \in {\mathfrak F}^{\rm LH}_U$ respectively. Then 
\[
\{ F^{V_1W_1},F^{V_2W_2}\} = F^{[V_1, W_1]\, [V_2, W_2]} ,  
\]
which implies that observable currents corresponding to the 
symplectic product of locally Hamiltonian vector fields 
form a Lie algebra isomorphic to the commutator subalgebra of the Lie algebra of locally Hamiltonian vector fields modulo gauge. 
\end{pro}

The result stated above together with other results of this section and Section \ref{HVFs} may be summarized in the following diagram. 
In the diagram 
${\rm SPLHOC}_U \subset {\rm HOC}_U$ denotes the Lie algebra of symplectic product observable currents associated to locally Hamiltonian vector fields 
as a sub Lie $n$-algebra of the algebra of Hamiltonian observable currents. 
We will state the result starting at a level which ignores the conditions on the boundary terms over $\partial U$. 
This occupies the two columns at the right of the diagram. 
Separately, in the two columns in the left,  we state the corresponding results which do take into account the conditions over $\partial U$, and the connection between the two sets of results 
are maps that exist in the scenario in which 
${\mathfrak G}_U$ is the maximal ideal in ${\mathfrak F}^{\rm LH}_U \cap \hat{\mathfrak G}_U$.%
\footnote{
An arrow starting with a hook denotes an inclusion map. An arrow with a double head at the end denotes an onto map. 
An arrow starting with a bifurcation ironically denotes a one to one map. 
} 
\[\xymatrix{
	\left[	{\mathfrak F}^{\rm LH}_U, {\mathfrak F}_U^{\rm LH}\right]/{\mathfrak G}_U
		\ar@{^{(}->}[r]
	&
	{\mathfrak F}^{\rm LH}_U/{\mathfrak G}_U
		\ar@{>->}[r]^{\mbox{conj.}}
	&
	\hat{\mathfrak F}^{\rm LH}_U/\hat{\mathfrak G}_U
		\ar@{<-^{)}}[r]
	&
	\left[	\hat{\mathfrak F}^{\rm LH}_U, \hat{\mathfrak F}_U^{\rm LH}\right]/\hat{\mathfrak G}_U
	&	
	\\
	{\rm SPLHOC}_U
		\ar@{->}[u]		\ar@{^{(}->}[r]
	&
	{\rm HOC}_U
		\ar@{->>}[u]		\ar@{^{(}->}[r]
		&
	\widehat{\rm HOC}_U
		\ar@{->>}[u]    \ar@{<-^{)}}[r]
	&
	\widehat{\rm SPLHOC}_U
		\ar@{->}[u]
		&
	\\
	0
		\ar@{^{(}->}[u]
	&
	C
		\ar@{^{(}->}[u]			\ar@{^{(}->}[r]
	&
	\hat{C}
		\ar@{^{(}->}[u]
	&
	0
		\ar@{^{(}->}[u]	
	&
}\]
where the sets of constant observable currents are 
\[
\hat{C} = \{ F \in {\rm OC}_U : ({\sf d_v} F - {\sf d_h}\sigma^{F})|_{{\mathfrak F}_U , \mathcal{E}_{L, U}} = 0 \} , 
\]
\[
C = \{ F \in {\rm OC}_U : ({\sf d_v} F - {\sf d_h}\sigma^{F})|_{{\mathfrak F}_U , \mathcal{E}_{L, U}} = 0 \} , 
\]
where the boundary terms for the elements of $C$ are required to satisfy \\
$(({\sf d_v}\sigma^F - \lambda^F)|_{{\mathfrak F}_U, \mathcal{E}_{L, U}} )|_{\partial U} = 0$ for respective horizontally closed forms 
${\sf d_h} \lambda^F |_{{\mathfrak F}_U, \mathcal{E}_{L, U}}= 0$.

%%%%%%%%%%%%%%%%%%%%%%%%%%%%%

\section{Observable currents separate solutions modulo gauge}
\label{OCsSeparate}

%%%%%%%%%%%%%%%%%%%%%%%%%%%%

In previous work on the multisymplectic approach to classical field theory it is argued that the set of physical observables that can be obtained from the integration of conserved currents is very limited including almost nothing besides Noether currents (see for example \cite{Forger+Romero, Helein, kanatchikov1998canonical, Kijowski, Goldschmidt(1973)}). 
Here we defined the notion of observable currents, 
which differs from that used in previous works in two aspects: First, we allow for currents depending on arbitrarily high order derivatives of the field. Second, the currents that we consider may be defined only locally in the jet. 
In previous sections we exhibited 
the large family of symplectic product observable currents together with the corresponding locally defined physical observables. 
In order to be conclusive showing that observable currents are an interesting source of physical observables we 
(i) prove that 
if our domain is foliated by Cauchy surfaces with no boundary 
the algebra of 
Hamiltonian observable currents is capable of distinguishing between gauge inequivalent solutions, and 
(ii) we give supporting evidence to conjecture that for any spacetime domain 
gauge inequivalent solutions may be distinguished by means of 
observable currents. 
With the aim of 
making the task more transparent we will prove a local version of statement (i), 
and then discuss an extension of our argument to address more general cases. 
The key ingredient of our result was introduced in Remark \ref{EveryHVFhasOCs} showing that every locally Hamiltonian vector field has corresponding observable currents. 
This ingredient is complemented with the assumption that {\em given a solution and a variation of it obeying the linearized field equation, 
there is a locally Hamiltonian vector field inducing the given variation of the given solution which is defined at least in a neighborhood of the solution. 
This assumption is very mild in the presence of a Cauchy surface with no boundary, and it is somewhat less trivial for a general domain 
due to the conditions on boundary terms of equation (\ref{HOCs}) over $\partial U$. 
} 
\begin{tma}\label{SeparationOfSols}
Assume that the spacetime domain $U$ 
is endowed with a foliation by Cauchy surfaces with no boundary and that the assumption stated above holds. 
Consider any curve of solutions $\phi_t \in {\rm Sols}_U$ starting at $\phi_0 = \phi$ 
whose tangent vector at the initial solution does not correspond to a gauge vector field in the jet. 
Then there is a Hamiltonian observable current $F\in {\rm HOC}_U$ defined at least in a neighborhood of 
$j \phi (U)$ 
such that for any Cauchy surface 
$f_\Sigma \in {\rm HObs}_\Sigma$ 
distinguishes between $[\phi]$ and neighboring points in the curve 
of solutions modulo gauge $[\phi_t]$. 
\end{tma}
\begin{proof}
We will prove the theorem showing that there must be a locally defined Hamiltonian observable current $F\in {\rm HOC}_U$ such that 
$f_\Sigma$ satisfies 
\[
\frac{d}{dt}|_{t=0} f_\Sigma[\phi_t] \neq 0 . 
\]
Consider any Cauchy surface $\Sigma$ 
and $W \in {\mathfrak F}_U^{\mathfrak G}$ modeling the variation of the curve of solutions $\phi_t$ at the initial time. 
The assumption is that $W \in {\mathfrak F}_U^{\mathfrak G}$ is not a gauge vector field, 
and this implies that 
there is 
$V \in {\mathfrak F}_U^{\mathfrak G}$ such that 
$\iota_w \iota_v \omega_{L\Sigma} = \int_\Sigma j\phi^\ast \iota_W \iota_V \Omega_L \neq 0$. 
Now we appeal to the validity of our assumption stating that 
$V$ can actually be assumed to be a locally Hamiltonian 
defined at least in a neighborhood of $j\phi(U)$; that is, $V \in \mathfrak{F}_U^{\rm LH} \subset {\mathfrak F}_U^{\mathfrak G}$. 
By construction, we know 
that $\iota_W \iota_V \Omega_L$ is not horizontally exact. 
Thus, 
thanks to the result of Remark \ref{EveryHVFhasOCs}, 
there is a locally defined Hamiltonian observable current $F\in {\rm HOC}_U$ 
such that 
\[
\mathscr{L}_W F = \iota_W {\sf d_v} F = 
- \iota_W \iota_V \Omega_L - {\sf d_h}\iota_W \sigma^F . 
\]
It follows that $f_\Sigma$ distinguishes between $\phi$ and neighboring solutions in the curve 
of solutions modulo gauge 
$[\phi_t]$. 
\end{proof}

The proof given above appeals to the existence of a (locally defined) locally Hamiltonian vector field 
$V \in \mathfrak{F}_U^{\rm LH}$ with the needed properties, and which ends up being the Hamiltonian 
vector field associated to the observable current distinguishing solutions modulo gauge. 
Since $V$ models a vector field in the space of solutions it 
has to solve the linearized field equation in a neighborhood, which are non linear, 
it is expected that $V$ depends on the field and partial derivatives of arbitrarily high orders. 
This forces the corresponding observable current $F$ not to fit in a jet of a predefined finite order. 
This was the main reason for us to work in the infinite jet, apart from computational convenience 
and geometric clarity. 
In the context of General Relativity Anderson and Torre \cite{anderson1996classification}
proved that Hamiltonian vector fields associated to non trivial observables 
need to depend on partial derivatives of the field of infinitely high order.

\begin{remark}[General observables written as integrals of observable currents]
\label{GenObsFromOCs}
Not all observables $f: {\rm Sols}_U \to {\mathbb R}$ are of the type $f_\Sigma \in {\rm Obs}_\Sigma$ for some hypersurface. However, under the conditions of Theorem 
\ref{SeparationOfSols} general arguments show that 
a large class of local observables can be approximated with arbitrary precision by means of observables induced by Hamiltonian observable currents. 
Here we will not deal with approximations; we will show, under the same assumptions, 
how to write any physical observable in that class 
as the integral of a, locally defined, observable current in an exact way. 

Spacetime localized measurement is the primary source of observables with direct physical interest. 
Such observables are modeled as local functionals, $a_M[\phi] = \int_M A(j\phi)$ 
(integrals over spacetime of densities $A$ depending on the field and its partial derivatives with compact support; for a refinement of this notion see \cite{khavkine2015local}). 
Moreover, the Peierls bracket assigns locally Hamiltonian vector fields $V_A \in {\mathfrak F}_U^{\rm LH}$ to such spacetime densities 
\cite{Forger+Romero, Khavkine(2013)}. 

%We just proved that, at least in domains that have a foliation by Cauchy surfaces with no boundary, functions of the type 
%$a_U[\phi] = \int_U A(j\phi)$ may be approximated by observables associated to Hamiltonian observable currents $\{ f^i_\Sigma \}$. 
%It is natural to raise the issue of wether this approximation is compatible with the assignments of Hamiltonian vector fields; that amounts to comparing $V_A$ with the, presumably convergent, sequence of vector fields $\{ V_{F^i}\}$; this issue has not been addressed yet. 
%%
%Below we will show that, under the assumptions of Theorem \ref{SeparationOfSols}, observables of the type $a_U$ can be written as integrals of observable currents. This exact result implies the convergence of $\{ V_{F^i}\}$ to $V_A$. 

Since the observable $a_M$ has $V_A$ as associated locally Hamiltonian vector field, its change under a variation of the field is determined by $V_A$. 
Thus, $a_M$ is determined by $V_A$ up to a constant that may be fixed using an auxiliary solution $\phi_0 \in {\rm Sols}_M$.%
\footnote{
We assume that the locally Hamiltonian vector field can be modeled as an evolutionary vector field 
$V_A \in {\mathfrak F}_U^{\rm LH}$. The information that we have concerns only $V_A|_{\mathcal{E}_{L, U}}$ the assumption consists merely on the extendibility of $jV_A$ to a neighborhood of $j\phi_0(U)$. 
} 
On the other hand, we showed that every locally Hamiltonian vector field has corresponding locally defined 
observable current (which under the current assumptions are Hamiltonian). Let $\tilde{F}_{V_A} \in {\rm OC}_U$ be one of the observable currents 
defined in a neighborhood of $j\phi_0(U)$ and 
satisfying 
$({\sf d_v} \tilde{F}_{V_A} + \iota_{jV_A} \Omega_L 
- {\sf d_h} \sigma^{\tilde{F}_{V_A}} )|_{{\mathfrak F}_U , \mathcal{E}_{L, U}} = 0$. 
The observable current that we are looking for is 
\[
F^{V_A} = \tilde{F}^{V_A} + F_0 , 
\]
where $F^0$ is a field independent current such that $f_\Sigma^{V_A}[\phi_0] = a_M [\phi_0]$ for any Cauchy surface $\Sigma$.

A direct treatment of localized observables in a covariant field theory formalism based on the variational bicomplex is given by Barnich et al \cite{barnich1998sh}. It would be interesting to explore the relation between their formalism and ours.

Observables that are of special importance for 
Proposition \ref{BracketObs} are the products of observables in ${\rm HObs}_\Sigma$. The observable $(fg)_\Sigma$ is defined as the product of the evaluations point wise in ${\rm Sols}_M$. Consider the observables $f_\Sigma$, $g_\Sigma$ and write them as integrals of (limits of) local functionals $a_M^F$, $b_M^G$. It is clear that the observables $(fg)_\Sigma : {\rm Sols}_M \to \mathbb R$ and $(ab)_M : {\rm Sols}_M \to \mathbb R$ are equal; then 
$(fg)_\Sigma$ must have an associated locally Hamiltonian vector field equal to that of $(ab)_M$. 
Product observables, those of the type $(ab)_M$ are not local functionals; however, they 
have been extensively studied and they do have associated Hamiltonian vector fields. 
Since observables corresponding to products of local functionals have associated locally Hamiltonian vector fields, 
the proof given above can be adapted to show that those observables can also be written exactly as 
integrals of observable currents. Therefore, observables of the type $(fg)_\Sigma$ 
can be written as integrals of observable currents. 
Thus, under the conditions of Proposition \ref{BracketObs}
${\rm HObs}_\Sigma$ has a Poisson algebra structure. 
Moreover, in this case formulations of field theory in terms of initial data have been thoroughly studied, and the algebra of observables in these approaches is generated by observables in ${\rm HObs}_\Sigma$. 

For a rigorous treatment focussing on gauge invariant conserved currents see Vitagliano's work on the covariant phase space \cite{Vitagliano(2009)}. 
\end{remark}

The previous theorem relies on the strong simplifying assumption of $U$ being foliated by 
Cauchy surfaces with no boundary. This assumption 
helped us 
avoid dealing with the condition on locally Hamiltonian vector fields $V \in \mathfrak{F}^{\rm LH}_U$ 
to yield a boundary term (in $(\mathscr{L}_V \Omega_L  - {\sf d_h} \sigma^V)|_{{\mathfrak F}_U , \mathcal{E}_{L, U}} = 0$) 
such that $(\sigma^V|_{{\mathfrak F}_U , \mathcal{E}_{L, U}})|_{\partial U} = 0$. 
Before addressing the extension of 
the result to regions without a foliation by Cauchy surfaces with no boundary 
there is an important cautionary remark to consider: 
We know that in the presence of a boundary there are non trivial observable currents that are not Hamiltonian according to our definition, a family of examples being the Noether currents associated to would be gauge symmetries. 
%The origin of this peculiarity is that would be gauge symmetries are associated to vector fields that 
%do not preserve the pre-symplectic forms $\omega_{L\Sigma}$ associated to hypersurfaces with $\partial \Sigma \subset \partial U$. 
Thus, there are observable currents which can not be approximated by Hamiltonian observable currents. 
We do not claim that in the presence of boundaries Hamiltonian observable currents do not distinguish gauge inequivalent solutions, but we know that there are some aspects of the field that are naturally measured by observable currents 
which are not Hamiltonian. 

\begin{remark}[In the absence of a Cauchy surface]
Now we work on a spacetime domain $U$ that may not have a Cauchy surface, and 
consider the 
issue of distinguishing gauge inequivalent solutions 
of a field theory with local degrees of freedom 
by means of observable currents (even if they are not Hamiltonian). 

Consider a curve of solutions $\phi_t$ as in Theorem \ref{SeparationOfSols}, again assuming that the initial velocity of the curve is 
modeled by a non trivial $W \in {\mathfrak F}_U^{\mathfrak G}$. Also consider a point in the interior of our spacetime domain 
$x \in U$ with a neighborhood around it $x \in U' \subset U$ 
in which $(\iota_{jW} \Omega_L)|_{{\mathfrak F}_U , \mathcal{E}_{L, U}}$ 
is not horizontally exact. 
Now assume that the field theory under consideration having local degrees of freedom means that 
in the situation just described 
there must be an oriented 
hypersurface $\Sigma$ intersecting $U'$ (with $\partial \Sigma \subset \partial U$, and in which we have chosen an auxiliary volume element), 
a localized variation of the field modeled by some $V \in {\mathfrak F}_U^{\rm LH}$ and a neighborhood $U'' \subset U'$ of $x$ 
such that: \\
(i) $((\iota_{jV} \iota_{jW} \Omega_L)|_{\mathcal{E}_{L, U}})|_{U''} = \lambda \pi^\ast \mbox{vol}_\Sigma$ 
for a function $\lambda$ locally defined in the jet that is {\em strictly positive} when restricted to the 
intersection of its domain of definition (which is an open set containing $j\phi(U')$) 
with the jet bundle over $U''$, 
and where we have written $\pi^\ast \mbox{vol}_\Sigma$ for $\Sigma$'s volume element pulled back to the infinite jet. \\
(ii) $jV|_{j\phi(\Sigma)}$ vanishes outside $j\phi(U'\cap \Sigma)$. 

From these assumptions it follows that for an observable current with 
$({\sf d_v} F + \iota_{jV} \Omega_L 
- {\sf d_h} \sigma^F )|_{{\mathfrak F}_U , \mathcal{E}_{L, U}} = 0$, 
whose existence is guaranteed by the result in Remark \ref{EveryHVFhasOCs}, 
we have 
\[
\frac{d}{dt}|_{t=0} f_\Sigma[\phi_t] =  
\int_\Sigma j\phi^\ast \mathscr{L}_W F = 
\int_\Sigma j\phi^\ast 
\iota_{jV} \iota_{jW} \Omega_L   > 0 , 
\]
where 
we have not written a boundary term for $\mathscr{L}_W F$ because its contribution to the above calculation vanishes. 
The reasons are the following: 
First, 
$jV$ vanishing in a neighborhood of the jet over $\partial \Sigma$ and 
$(\mathscr{L}_{jV} \Omega_L - {\sf d_h} \sigma^V )|_{{\mathfrak F}_U , \mathcal{E}_{L, U}} = 0$ 
imply that $(j\phi^\ast \sigma^V)|_{{\mathfrak F}_U} $ vanishes over $\partial \Sigma$. 
Second, the boundary terms associated to $V$ and $F$ are related by 
$\sigma^V = - {\sf d_v} \sigma^F + \lambda^F$ (with $\lambda^F$ horizontally closed) ; thus, if we choose an observable current 
without a vertically constant 
${\sf d_h}$-exact component the boundary term of $\mathscr{L}_W F$ would not contribute to the above calculation. 
Notice that even when the perturbation is localized in $\Sigma$, 
since $V$ follows the linearized field equation in general $V$ and $\sigma^V$ 
do not vanish everywhere on the jet over $\partial U$ and in general this construction does not provide Hamiltonian observable currents. 

Given $[V] \in {\mathfrak F}_U^{\rm LH} / {\mathfrak G}_U$ and $F \in {\rm OC}_U$ with 
$({\sf d_v} F + \iota_{jV} \Omega_L 
- {\sf d_h} \sigma^F )|_{{\mathfrak F}_U , \mathcal{E}_{L, U}} = 0$ for some $V \in [V]$ the problem of finding 
a Hamiltonian observable current $F \in {\rm HOC}_U$ associated to  $[V]$ amounts to 
finding the appropriate ``may be gauge vector field'' 
$X \in \hat{\mathfrak G}_U$ such that $(\sigma^X - \sigma^V)$ satisfies the requested conditions over $\partial U$. 
\end{remark}

\section{Observable current spaces of nested and glued domains}
\label{NestedAndGluedDomains}

%%%%%%%%%%%%%%%%%%%%%%%%%%%%

Let us start with the case of a domain contained in another one $U' \subset U$. The space ${\rm OC}_U$ is composed by differential forms 
and it hosts an equivalence relation stating that observable currents differing by 
${\sf d_h} \sigma = F - G$ with $(\sigma|_{\mathcal{E}_{L, U}})|_{\partial U} = 0$ 
are physically equivalent; this equivalence relation may be called $U$-restricted ${\sf d_h}$-cohomology. 
Consider $F, G \in {\rm OC}_U$ which are equivalent under $U$-restricted ${\sf d_h}$-cohomology and restrict them to the smaller domain $U'$: 
We will find that the conserved currents $F|_{U'}, G|_{U'}$ are not necessarily equivalent according to $U'$-restricted ${\sf d_h}$-cohomology. 
Additionally, they may not be gauge invariant in $U'$ because the restriction of some gauge vector fields from 
${\mathfrak G}_{U}$ to $U'$ do not belong to ${\mathfrak G}_{U'}$ because of not satisfying Condition 2 over $\partial U'$. 
Thus, the restriction of observable currents from $U$ to $U'$ is not a natural operation.

In the case of spacetime {\em localized observables} there is a natural map from the space of observables corresponding to the smaller domain to the space of observables corresponding to the bigger domain; 
in that case the densities modeling the observable in the smaller domain are simply extended by zero to become defined in the larger domain. 
This simple extension does not work for observable currents because they need to satisfy a conservation law. 
In Remark \ref{GenObsFromOCs} we briefly commented on observable currents induced by localized measurements and its relation to Peierls' bracket. In the context of that procedure, observable currents on the bigger domain can be induced by localized measurements in the smaller domain 
as expected for observables associated to localized measurements 
\cite{costello2016factorization}.  

Now consider a domain that is composed by two subdomains intersecting along a hypersurface, $U = U_1 \#_\Sigma U_2$. 
We will see how compatible observable currents of the subdomains produce an observable current in ${\rm OC}_U$. 

\begin{dfn}[Gluing algebras of adjacent domains] 
The construction needs the following definitions:  
\begin{itemize}
\item
$
{\rm OC}_{U_1} \#_\Sigma  {\rm OC}_{U_2} = 
\left\{ 
( F_1, F_2 ): F_i \in {\rm OC}_{U_i} \mbox{ with } F_1|_\Sigma =  F_2|_\Sigma \mbox{ and } {\sf d_v} F_1|_\Sigma =  {\sf d_v} F_2|_\Sigma
\right\}
$. 
\item
${\mathfrak G}^{\hat{\Sigma}}_{U_i}$ is the subalgebra of ${\mathfrak F}_{U_i}$ whose elements satisfy Condition 1 for gauge vector fields and a weaker version of Condition 2: 
$X \in {\mathfrak G}^{\hat{\Sigma}}_{U_i}$ if and only if 
$(\iota_{jX} \; \Omega_L)|_{{\mathfrak F}_{U_i} , \mathcal{E}_{L, U}}$ is horizontally exact, 
and $jX$ vanishes on the intersection of $\mathcal{E}_{L, U}$ with the bundle over $\partial U_i \setminus \Sigma^\circ$. 
\item
$
{\mathfrak G}^{\hat{\Sigma}}_{U_1} \#_\Sigma {\mathfrak G}^{\hat{\Sigma}}_{U_2} = 
\left\{ 
(V_1, V_2): V_i \in \hat{{\mathfrak G}}_{U_i} \mbox{ with } 
j_\Sigma (V_1|_{\mathcal{E}_{L, U} , \Sigma}) = j_\Sigma (V_2|_{\mathcal{E}_{L, U} , \Sigma}) 
\right\}
$. \\
$0$th order continuity of $(V_1, V_2)$ along the intersection of 
jet bundle over $\Sigma$ with $\mathcal{E}_{L, U}$ is equivalent to 
$j_\Sigma (V_1|_{\mathcal{E}_{L, U} , \Sigma}) = j_\Sigma (V_2|_{\mathcal{E}_{L, U} , \Sigma})$, where $j_\Sigma$ denotes the prolongation in which partial derivatives are calculated with respect to a local coordinate chart tailored to $\Sigma$ and partial derivatives in directions normal to $\Sigma$ are not considered. 
The linearized gluing field equation is trivially satisfied. 
Notice that ${\mathfrak G}_U$ is naturally injected into 
${\mathfrak G}^{\hat{\Sigma}}_{U_1} \#_\Sigma {\mathfrak G}^{\hat{\Sigma}}_{U_2}$. 
\\
$
{{\mathfrak G}}_{U_1} \#_\Sigma {{\mathfrak G}}_{U_2}$ denotes simply pairs of elements of ${{\mathfrak G}}_{U_i}$. 
\item
${\mathfrak G}_\Sigma = \frac{{\mathfrak G}^{\hat{\Sigma}}_{U_1} \#_\Sigma {\mathfrak G}^{\hat{\Sigma}}_{U_2}
}{{\mathfrak G}_{U_1} \#_\Sigma {\mathfrak G}_{U_2}
}$. 
\item
${\rm Inv}_{{\mathfrak G}_\Sigma} \! \left( {\rm OC}_{U_1} \#_\Sigma  {\rm OC}_{U_2} \right)$ 
denotes the subspace of 
${\rm OC}_{U_1} \#_\Sigma  {\rm OC}_{U_2}$ 
that is invariant under 
${\mathfrak G}_\Sigma$. 
\end{itemize}
\end{dfn}
The following proposition follows from the definitions. 
\begin{pro}
\[
{\rm OC}_U = {\rm Inv}_{{\mathfrak G}_\Sigma} \! \left( {\rm OC}_{U_1} \#_\Sigma  {\rm OC}_{U_2} \right) . 
\]
\end{pro}

Now consider the situation in which a domain with a Cauchy surface is divided into two subdomains $U = U_1 \#_{\Sigma} U_2$ in such a way that the Cauchy surface is also subdivided by a codimension two surface $\Delta$ as 
$S = S_1 \#_\Delta S_2$. 
Due to Condition 2 in the definition of gauge vector fields we have a subalgebra of observables associated to each of the portions of Cauchy surface ${\rm Obs}_{S_i}$; moreover it is clear from the definitions of the observables that any element of ${\rm Obs}_{S}$ can be written as a sum of two terms $f_S = f_{S_1} + f_{S_2}$ belonging to ${\rm Obs}_{S_i}$. Thus, ${\rm Obs}_{S}$ is recoverable from ${\rm Obs}_{S_1}$ and ${\rm Obs}_{S_2}$.

%%%%%%%%%%%%%%%%%%%%%%%%%%%%

\section{Example: Maxwell field}
\label{Ex}

%%%%%%%%%%%%%%%%%%%%%%%%%%%%

In this section we give the notation and initial 
setup to treat the Maxwell field in this formalism. 
The presentation is not pedagogical; the aim of this section is only to be used as a reference for the reader to be able to work on this familiar example by him self or her self. 
We also mention particularly illustrative results that are easily obtainable in this prime example of a linear gauge field theory. 

The notation for the general case is given in the appendix; in this section we follow that notation only in its essence. 
In the general case a field is denoted by $\phi^a$, and 
partial derivatives in a coordinate chart are written as $\partial_i \phi = \frac{\partial \phi^a}{\partial x^i}$. In this example the field is taken to be the potential one form $A$. 
%When working on a chart, 
%our notation for the field will be $A_\nu$. 

Let $M = {\mathbb R}^4$ with the Minkowski metric $\eta$. 
Histories, i.e. local sections, are one forms; then $Y = T^\ast {\mathbb R}^4$. The notation for elements in the first jet will be 
$j^1A(x) = (x^\mu ; A_\nu (x) ; v_{\nu \mu }= \partial_\mu A_\nu(x)) \in J^1Y$. A general point in the infinite jet will be denoted by 
\[
( x^\mu ; A_\nu ; v_{\nu \mu } ; v_{\nu \mu \rho} ; \ldots ) . 
\]
In the Lagrangian density only the skew symmetric combination 
$F_{\mu \nu} = v_{\nu \mu } - v_{\mu \nu }$ appears, 
\[
L = \frac{-1}{4} F_{\mu \nu} F^{\mu \nu} d^4 x . 
\]
Basic vector fields in the infinite jet are denoted by  
$ \{ \partial_\mu= \frac{\partial}{\partial x^\mu} ; 
\partial_A^\nu= \frac{\partial}{\partial A_\nu} ; 
\partial^{\nu \mu} = \frac{\partial}{\partial v_{\nu\mu}} ; \ldots \}$. 
The generators of the exterior algebra of differential forms in the infinite jet are 
$ \{ d x^\mu ; 
\vartheta_\nu = d A_\nu -  v_{\nu \mu }d x^\mu 
; 
\vartheta_{\nu \mu} = dv_{\nu \mu}- v_{\nu \mu \rho}d x^\rho 
; \ldots \}$.

The non zero horizontal differentials of the coordinates and basic forms are:\\ 
$
\mathsf{d_h} x^\mu = d  x^\mu ; 
\mathsf{d_h} A_\nu = v_{\nu \mu }d x^\mu ; 
\mathsf{d_h} v_{\nu \mu } = v_{\nu \mu \rho}d x^\rho ; 
\ldots ;
\mathsf{d_h} \vartheta_\nu = d x^\mu \wedge \vartheta_{\nu \mu} ; 
\ldots ;
\mathsf{d_h} F_{\mu \nu} = D_\rho F_{\mu \nu} dx^\rho 
= ( v_{\nu \mu \rho} -v_{\mu \nu \rho}) dx^\rho
$. 
The non zero vertical differentials are:\\ 
$
\mathsf{d_v} A_\nu = \vartheta_\nu ; 
\mathsf{d_v} v_{\nu \mu } = \vartheta_{\nu \mu} ; 
\ldots
\mathsf{d_v} F_{\mu \nu} = \vartheta_{\nu \mu} - \vartheta_{\mu \nu} ; 
\mathsf{d_v} L = \frac{-1}{2} (\vartheta_{\nu \mu} - \vartheta_{\mu \nu}) F^{\mu \nu} d^4 x 
$. 

From the equation $\mathsf{d_v} L = I (\mathsf{d_v} L) + \mathsf{d_h} \Theta_L$ follows that 
the field equation is \\
$jA^\ast (I (\mathsf{d_v} L))= 0$,  where 
\[
I (\mathsf{d_v} L) = \frac{1}{2}\vartheta_\sigma \wedge 
D_\rho [ \iota_{\partial^{\sigma \rho}} (\vartheta_{\nu \mu} - \vartheta_{\mu \nu}) F^{\mu \nu} ] d^4 x
= ( v^{\nu \mu}_\mu - v^{\mu \nu}_\mu ) \vartheta_\nu \wedge d^4 x .   
\]
Simple substitution shows that the usual Maxwell field equation is recovered. 

From the same equations 
it is easy to verify that a pre-multisymplectic potential that works is 
$\Theta_L = F^{\mu \nu} \vartheta_\mu \wedge d^3x_\nu$; which yields 
\[
\Omega_L = - \mathsf{d_v}\Theta_L 
= (\vartheta^{\mu \nu} - \vartheta^{\nu \mu} ) \wedge \vartheta_\mu \wedge d^3x_\nu . 
\]

In our framework field perturbations play a central role. 
In linear field theories generic perturbations 
correspond to one parameter families of solutions of the type 
$A_\nu(t) = (A_\nu + t V_\nu)$, where both $A_\nu$ and $V_\nu$ are solutions. The corresponding evolutionary vector field 
may be written as 
$V= V_\nu \partial_A^\nu$, and its prolongation to the infinite jet is 
\[
jV = V_\nu \partial_A^\nu + \partial_\mu( V_\nu ) \partial^{\nu \mu} + \ldots . 
\]
The field perturbation $X^f$ corresponding to $A_\nu(t) = A_\nu + t \partial_\nu f$ is 
$X^f= \partial_\nu f \partial_A^\nu$, and its prolongation to the infinite jet is 
\[
jX^f= \partial_\nu f \partial_A^\nu + \partial_\mu( \partial_\nu f ) \partial^{\nu \mu} + \ldots . 
\]
If the function which determines the vector field is the pullback of a spacetime function $f = \pi^\ast \tilde{f}$ with $\tilde{f}: U \to {\mathbb R}$ proving that 
$(\iota_{jX^f} \Omega_L)|_{{\mathfrak F}_U , \mathcal{E}_{L, U}}$ is $\mathsf{d_h}$-exact is not a difficult exercise. 
Moreover, after a careful study, one can verify that 
Condition 1 for gauge vector fields implies that the perturbation must come from an exact one form.

% The sentence below was in the same paragraph as the text above. 
A new element is Condition 2 for gauge vector fields; for a field independent perturbation of the type described above 
$X^f$ to be gauge the locality condition demands that $\partial_\nu f|_{\partial U}=0$, 
and that all the higher order partial derivatives also vanish over $\partial U$.

On the other hand, a perturbation $V_\nu$ of the type written above is not trivial. 
Due to its independence of the field  ($\mathsf{d_v} V\nu=0$), it is simple to prove that 
$\mathscr{L}_{jV} \Omega_L = 0$ showing that $V_\nu$ is a locally Hamiltonian vector field. 

If we use two perturbations $V$ and $W$ of the type exhibited above 
we can write their symplectic product observable current $F^{V W}$. The result is the simplest observable current --a constant current--; when integrated on a hypersurface it yields a constant function. We could do the calculation directly, but one can also notice that the Hamiltonian vector field associated to $F^{V W}$ is $[V, W]=0$ which the reader may verify from the definition of these vector fields. 

A slightly more general type of perturbation depending on the field can be constructed using the linearity of the fibers
in the bundle. The perturbation over each spacetime point 
may be linear functions of the field $V'_\nu = M_\nu ^\mu A_\mu + V_\nu$; 
one discovers that if the matrix is constant, the field solves the field equation, and $V_\nu$ solves the 
linearized field equation then $V'_\nu$ also solves the linearized field equation. 
The observable currents $F^{V' W'}$ are slightly less trivial for this family of vector fields.

We can also look for a Hamiltonian observable current with a given 
$V$ as associated vector field. Since the Maxwell field is linear, the observable current we are looking for is $F^{V} \in {\rm OCs}_U$ as defined in Remark \ref{LineFieldsOCs}.

For an example that exhibits abelian gauge freedom and nonlinearities it may be a good idea to explore the Born-Infield model.

%%%%%%%%%%%%%%%%%%%%%%%%%%%%

\section*{Appendix: Minimal set of definitions about the 
variational bicomplex}
\label{VarBiC}

%%%%%%%%%%%%%%%%%%%%%%%%%%%%

This minimalistic revision of the variational bicomplex may serve the purpose of letting someone that knows another presentation of classical field theory, like the covariant phase space formalism, read this article. For an introductory presentation of the ideas of the subject the reader is referred to Anderson's brief introduction \cite{anderson1992introduction}. 

Let $M$ be an $n-$dimensional manifold and $\pi:Y\rightarrow M$ be a fiber bundle with $m-$dimensional fiber $F$. 
%Denote its sections or {\em histories} as $\mathrm{Hists}_{U}:=\Gamma(\pi\mid_U)$ where $U\subseteq M$ is a compact domain with piecewise smooth boundary.

Points in the $k-$jet bundle $\pi_{k,0}:J^kY\rightarrow Y$, $k=1,2,\dots$ correspond to 
equivalence classes of local sections of $\pi$ that agree up to $k$-th order partial derivatives when evaluated at a given point $x\in M$. 
If in the restriction of $Y$ over a coordinate chart of the base 
$U\subset M$ we use coordinates such that the evaluation of a local section is 
$\phi(x) = (x^1,\dots,x^i,\dots,x^n; u^1,\dots,u^a,\dots,u^m) \in Y|_U$, then 
we get the following coordinates for the $k$-jet 
\[
	\left(x;u^{(k)}\right):=
		(x^1,\dots,x^i,\dots,x^n;u^1,\dots,u^a,\dots,u^m;\dots,u^{a}_I,\dots) ,
			\in J^k Y |_U
\]
where $ i=1,\dots, n;\, a=1,\dots,m;$ and $I=(i_1,\dots,i_k)$ denotes a {\em multiindex} consisting on an unordered $k$-tuple of coordinate indices (the type of indices indicating one $k$-th order partial derivative of a function among all the possible $k$-th order partial derivative of that function). 
The degree of the multiindex is $|I|:=i_1+\dots+i_n=0,1,\dots,k$, $i_j\geq 0,i_j\in\mathbb{N}$. For $I=\emptyset $, we set $u_\emptyset^a=u^a$. 

The projection $\pi_{k+r,k}:J^{k+r}Y\rightarrow J^kY$ 
is defined by forgetting the coordinates corresponding to partial derivatives of higher order. 
The infinite jet $J^\infty Y$ may be defined as the inverse limit of this system of projections, and it is 
the space where the formalism of the variational bicomplex takes place. For notational convenience we  
denote it simply by $JY$. 
The jets of finite order can be thought of as truncations of it corresponding to neglecting all the partial derivatives of order higher than what a certain cut-off specifies. 

For a local section $\phi: U \subset M \rightarrow Y|_U$, its prolongation to the $k-$jet $j^k\phi:U \subset M \rightarrow J^k Y|_U$ is the section 
\[
	j^k\phi(x)=
		\left(x^1,\dots,x^i,\dots,x^n;
		\phi^1(x),\dots,\phi^m(x);\dots,
		\frac{\partial^{|I|} \phi^{a}}{\partial^{i_1} x^1 \dots \partial^{i_n} x^n}(x),\dots\right) ,
\]
where $k$ may be taken finite or $k = \infty$ . 

Contact forms are differential forms $\omega$ in $JY$ such that their pull back to the base by the prolongation of any local section vanishes, 
$j\phi^\ast \omega = 0$. 

The exterior algebra of differential forms in $J Y$ is generated by the set of one forms $\{ dx^i , \vartheta^a_I \}$, where 
\[
	\vartheta^a_I:=
		du^a_I-\sum_{j=1}^nu^a_{(I,j)}dx^j , 
\]
and the the set $\{ \vartheta^a_I \}$ generates the ideal of contact forms. 

Vector fields in $Y$ can be promoted to vector fields in $JY$; 
since this construction is related to the prolongation of sections, 
the technical term is prolongation of vector fields. 
Let $V$ be a vector field in $Y$. Its unique prolongation $jV$ to $JY$ is the 
vector field% 
\footnote{
In many references the prolongation is denoted by $\mbox{pr} V$. 
}
which (i) agrees with $V$ when differentiating functions of $Y$, and (ii) which preserves the contact ideal. The geometric motivation of this condition is generating a flow sending sections which are prolongations to other sections of the same type. 

Since the main focus on field theory is not the points of $Y$, but its local sections, the objects generating flows of sections are of primary importance and they are not vector fields in $Y$. First of all the flow has to send local sections over $U$ to other local sections over $U$; this property is fulfilled by vertical vector fields, but the space of vertical vector fields is not large enough to generate general local flows of sections. The local evolution laws for sections that are appropriate for field theory are captured by the so called {\em evolutionary vector fields}. They may be thought of as ``vector fields'' in $Y$ whose coefficients may depend on partial derivatives of the field of arbitrarily high order; equivalently they can be seen as vector fields in $JY$ which have been truncated to include only their components in $Y$. In terms of the coordinates chose above they are written as 
\[
V = 0 \frac{\partial}{\partial x^i} + V^a \frac{\partial}{\partial u^a} . 
\]
Evolutionary vector fields also have prolongations to $JY$ enjoying the properties of prolongations described in the previous paragraph. 
Their explicit form is 
\[
jV = \sum_{|I|=0}^\infty (D_I V^a) \frac{\partial}{\partial u^a_I} ,  
\]
where $D_i$ is the total derivative defined below and $D_I$ means to successively apply it according to the multi index $I$.

The flow in the space of sections generated by an evolutionary vector field $V$ sends sections that are prolongations to other sections that are prolongations if $jV$ preserves the contact ideal. This could be written as the condition demanding that for any $\vartheta^a_I$ its Lie derivative 
$\mathscr{L}_{jV} \vartheta^a_I$ is a contact form. 

It is possible to define a bracket for evolutionary vector fields which makes the space of such objects a Lie algebra; the main property of this definition is that for any evolutionary vector fields $V, W$ we have $j[V, W] = [jV, jW]$. For a more thorough explanation see \cite{Vinogradov}.

A general $p$-form is written as a sum of terms with products of 
$p$ one forms among $\{ dx^i , \vartheta^a_I \}$; factors of the type $dx^i$ are called  ``horizontal'', factors of the type 
$\vartheta^a_I$ are called ``vertical''. Thus, the space of $p$-forms becomes a direct sum of spaces ${\mathsf{\Omega}}^{r,s}(J^\infty Y)$ of forms which are products of exactly $r$ horizontal one forms and $s$ vertical one forms. 
The differential brings up the degree of forms by one and the direct sum structure mentioned makes the differential split as a sum of operators 
\[
\mathsf{d}= \mathsf{d_h} + \mathsf{d_v} , 
\]
where 
$\mathsf{d_h}: {\mathsf{\Omega}}^{r,s}(J^\infty Y) \to {\mathsf{\Omega}}^{r+1,s}(J^\infty Y)$ and 
$\mathsf{d_v}: {\mathsf{\Omega}}^{r,s}(J^\infty Y) \to {\mathsf{\Omega}}^{r,s+1}(J^\infty Y)$ are characterized by their action on functions 
\[
	\mathsf{d_h} f = 
		\left( \frac{\partial f}{\partial x_i}
		+
			u^a_{(J,i)}\frac{\partial f}{\partial u^a_J} \right) dx^i
			= (D_i f) dx^i , \quad \quad
			\mathsf{d_v}f
	=
		\frac{\partial f}{\partial u^a_I}\vartheta^a_I .
\]
For the generating one forms we get 
\[
\mathsf{d_h}dx^i = 0 , \quad \mathsf{d_v}dx^i = 0 , \quad 
\mathsf{d_h}\vartheta^a_I = dx^i \wedge \vartheta^a_{(I,i)} , \quad 
\mathsf{d_v}\vartheta^a_I = 0 . 
\]
The following identities hold 
\[
	\mathsf{d_h}^2=0 ,\qquad 
	\mathsf{d_v}\mathsf{d_h}= - \mathsf{d_h}\mathsf{d_v} ,\qquad 
	\mathsf{d_v}^2=0	.
\] 
Other identities that we use repeatedly are 
\[
\iota_X \mathsf{d_h} F = - \mathsf{d_h} \iota_X  F  , \quad 
j\phi^\ast \mathsf{d_h} F = d \, j\phi^\ast F , 
\]
where $X$ is any evolutionary vector field and $\phi$ is any section of $Y$. 

The field equation is $j\phi^\ast E(L)=0$, where ${E}(L)$
appears in the derivative of the Lagrangian density (which according to the terminology just given is a form of horizontal degree $n$ and vertical degree zero),  
$\mathsf{d_v} { L} = {E}(L) + \mathsf{d_h} \Theta_L$. From our first encounter with the Euler-Lagrange equation we know that integration by parts is an essential step in its derivation. In the language of the variational bicomplex the definition is 
\[
{E}(L) = \mathsf{I} \mathsf{d_v} { L} , 
\]
where the integration by parts operator 
$\mathsf{I}: {\mathsf{\Omega}}^{n,s}(J^\infty Y) \to {\mathsf{\Omega}}^{n,s}(J^\infty Y)$ is defined by 
\[
\mathsf{I} = \frac{1}{s} \vartheta^a F_a , \quad  
F_a \mu = \sum_{|J|}^k  
{\rm sgn}(|J|) D_J \iota_{\partial_a^J} \mu , 
\]
where 
$\partial_a^J  = \frac{\partial}{\partial u^a_J}$, 
${\rm sgn}(|J|)$ is positive for $|J|$ even, and 
the sum stops at $k$ if $\mu$ fits in $J^kY$ (i.e. if $\mu$ is a differential form of order $k$). The integration by parts operator $\mathsf{I}$ and $F_a$ have the following properties 
\[
F_a \circ \mathsf{d_h} = 0 , \quad 
\mu = \mathsf{I}(\mu) + \mathsf{d_h} \eta , \quad
\mathsf{I}^2 = \mathsf{I} . 
\]
The differential operators $\mathsf{d_h} , \mathsf{d_v}$ 
among the spaces ${\mathsf{\Omega}}^{r,s}$ 
are complemented by the map $\mathsf{I}$ and the spaces 
${\hat{\mathfrak F}}^{s} = \mathsf{I}({\mathsf{\Omega}}^{n,s})$ and the maps 
$E = \mathsf{I} \mathsf{d_v}: {\mathsf{\Omega}}^{n,0} \to {\hat{\mathfrak F}}^{1}$, 
$\delta = \mathsf{I} \mathsf{d_v}: {\hat{\mathfrak F}}^{s} \to {\hat{\mathfrak F}}^{s+1}$ to form an 
augmented variational bicomplex. The 
Euler-Lagrange complex resides at the corner of the augmented variational bicomplex starting at the spaces ${\mathsf{\Omega}}^{r,0}$ moving with $\mathsf{d_h}$ and then turning with $E$ to the spaces ${{\mathfrak F}}^{s}$ and moving with the differential $\delta$. 

In our definition of gauge vector fields the multisymplectic form $\Omega_L$ plays an essential role. In the context in which it appears, the gluing field equation, it is natural to consider it as restricted to a hypersurface and integration by parts becomes necessary to obtain the gluing field equation. Thus, in a slight abuse of notation we give the name $\mathsf{I}$ to 
the operator 
$\frac{1}{s} \vartheta^a F_a: {\mathsf{\Omega}}^{n-1,1}(J Y) \to {\mathsf{\Omega}}^{n-1,1}(J Y)$. 
%The gluing field equation at a hypersurface $\Sigma$ becomes 
%$0= I\Omega_L |_\Sigma = - I \mathsf{d_v} \Theta_L |_\Sigma$. 

%Additionally, Helmholtz conditions for the inverse problem of the calculus of variations on a differential operator $\Delta$ that may be associated to a form of degree $(n, 1)$ 
%are satisfied if and only if $I(\mathsf{d_v} \Delta)=0$. 

In a few instances during the article we alluded to ``Takens' Lemma'' 
\cite{Takens(1979)}. Here we state the part of the mentioned lemma that we need. 
\begin{lma}\label{lma:Takens}
For every ${\sf d_h}-$closed form $\tau \in {\sf \Omega}^{n-1,1}\left(  J^k Y  \right)$ there exists $\sigma\in {\sf \Omega}^{n-2,1}\left(J^rY\right)$, for some $r$ such that $\tau = {\sf d_h}\sigma$.
\end{lma}

For essential geometrical arguments that we did not give and for important features of the variational bicomplex that we did not cover (because they are not essential in this article) 
see Anderson's introduction \cite{anderson1992introduction}. 

We finish the appendix stating and proving 
a version of Noether's theorem.

\begin{tma}[Noether]
A Lagrange symmetry is one generated by an evolutionary vector field $V$ satisfying 
$
	\mathscr{L}_{jV}{{ L}} 
	= 
	{\sf d_h} \sigma_L^{V}
$. 
Every Lagrange symmetry has a corresponding Noether current 
$
	N^V
	=
	- \iota_{jV}\Theta_L + \sigma_L^{V} 
$ which satisfies the equation 
\[
({\sf d_v} N^V + \iota_{jV} \Omega_L - {\sf d_h}\sigma_N^V)|_{{\mathfrak F}_U, \mathcal{E}_{L, U}} =0
\]
for some differential form $\sigma_N^V$. \\
If $\sigma_N^V|_{\mathcal{E}_{L, U} , \partial U} = 0$ then 
$N^V \in {\rm HOC}_U$ with $V$ as its Hamiltonian vector field. 
\end{tma}
\begin{proof}
It is clear that, being a symmetry generator, $V$ preserves the variational principle and this implies that it is a solution of the linearized field equation. Moreover, we will see that $V$ satisfies a much stronger equation. 
Notice that $\mathscr{L}_{jV} E(L) = E( \mathscr{L}_{jV} L)$, and since the integration by parts operator annihilates horizontally exact forms we have that for Lagrangian symmetry generators 
$\mathscr{L}_{jV} E(L) = 0$. 

The proof that the current is conserved is trivial from the definition of $N^V$. 
The gauge invariance of the current follows from ${\sf d_v}N^V = -\iota_{jV} \Omega_L + {\sf d_h} \sigma^{N^V}$ for some differential form $\sigma^{N^V}$, which we will prove below. 

Direct calculation leads to 
${\sf d_v}N^V = -\iota_{jV} \Omega_L - \mathscr{L}_{jV} \Theta_L + {\sf d_v} \sigma_L^{V}$. 
Now we show that the last two terms correspond to a horizontally exact form. 
Observe that $\mathsf{d_v} { L} = {E}(L) + \mathsf{d_h} \Theta_L$ leads to 
$\mathscr{L}_{jV} E(L) = \mathsf{d_h} (- {\sf d_v} \sigma_L^{V} + \mathscr{L}_{jV} \Theta_L)$. 
Since in the first paragraph we saw that the right hand side is zero, 
we appeal to the lemma of Takens \cite{Takens(1979)} to guarantee the existence of a form $\sigma^{N^V}$ such that 
$-(\mathscr{L}_{V} \Theta_L - {\sf d_v} \sigma_L^{V}) = {\sf d_h} \sigma_N^{V}$. 
Thus according to our definition, if the boundary term satisfies the condition 
$\sigma_N^V|_{{\mathfrak F}_U , \mathcal{E}_{L, U}} = 0$ the current $N^V$ is a Hamiltonian observable current. 
\end{proof}

%%%%%%%%%%%%%%%%%%%%%%%%%%%%

\section*{Acknowledgements}

%%%%%%%%%%%%%%%%%%%%%%%%%%%%
We acknowledge correspondence and discussions about the subject of the article with 
Jasel Berra, Laurent Fridel, Igor Khavkine, Claudio Meneses, Alberto Molgado, Robert Oeckl, Michael Reisenberger and Jos\'e A. Vallejo. 
This work was partially supported by grant PAPIIT-UNAM IN 109415. 
JAZ was supported by a sabbatical grant by PASPA-UNAM.

%\printbibliography

%\bibliography{OCsBiblio}{}
%\bibliographystyle{plain}

%\bibliographystyle{unsrt}
%\bibliography{OCsBiblio}

%\bibliography{OCsBiblio}{}
%\bibliographystyle{plain}

\bibliography{OCsBiblio} 
\bibliographystyle{unsrt}

\end{document}